\documentclass[letterpaper,11pt,twoside,keywordsasfootnote,addressatend,noinfoline]{article}

\usepackage{imsart}
\usepackage{graphicx}
\usepackage{subfigure}
\usepackage{color}
\usepackage{tikz}
\usepackage{newlfont}
\usepackage{multirow}
\usepackage{array}
\usepackage{longtable} 
\usepackage{ifthen}
\usepackage{alltt}
\usepackage{float}
\usepackage{enumerate}
 \usepackage[text={6.5in,8.5in},centering]{geometry} 
\newcolumntype{C}[1]{>{\centering\let\newline\\\arraybackslash\hspace{0pt}}m{#1}}
\newcolumntype{L}[1]{>{\raggedright\let\newline\\\arraybackslash\hspace{0pt}}m{#1}}
\newcolumntype{R}[1]{>{\raggedleft\let\newline\\\arraybackslash\hspace{0pt}}m{#1}}
\newcommand{\ba}{\begin{array} }
\newcommand{\ea}{\end{array} }
\newcommand{\bae}{\begin{eqnarray}}
\newcommand{\eae}{\end{eqnarray}}
\newcommand{\bea}{\begin{eqnarray*}}
\newcommand{\eea}{\end{eqnarray*}}
\newcommand{\be}{\begin{equation}}
\newcommand{\ee}{\end{equation}}

\newcommand{\pr}{{\bf Proof}~~}

\def\to{{\rightarrow}}
\def\qed{\mbox{}\hfill\raisebox{-2pt}{\rule{5.6pt}{8pt}\rule{4pt}{0pt}}\medskip\par}

\def\to{{\rightarrow}}

\usepackage{amsfonts}
\usepackage{amssymb}
\usepackage{amsthm}
\usepackage{ mathrsfs }
 \usepackage{ae} 
\usepackage[T1]{fontenc}
\usepackage[ansinew]{inputenc}
\usepackage{amsmath}

\usepackage[english]{babel}

\usepackage{graphicx}
\usepackage{color}
\usepackage[colorlinks]{hyperref}
\usepackage{lscape}
\usepackage{graphicx}
\usepackage{epstopdf}
\newtheorem{theorem}{\hskip\parindent\bf Theorem}[section]

\usepackage{multicol}
\usepackage{tabularx}
\usepackage{ctable}
\usepackage{booktabs}
\usepackage{bm}

\begin{document}

\begin{frontmatter}

\title  {$SIS$ and $SIR$ epidemic models under virtual dispersal}
\runtitle  {$SIS$ and $SIR$ epidemic models under virtual dispersal}
\author    {Derdei Bichara, Yun Kang, Carlos Castillo-Chavez, Richard Horan, Charles Perrings}
\runauthor {Derdei Bichara, Yun Kang, Carlos Castillo-Chavez, Richard Horan and Charles Perrings}

\address   {SAL Mathematical, Computational and Modeling Science Center, \\Arizona State University, Tempe, AZ 85287.\\ E-mail:  derdei.bichara@asu.edu}
\address   {Sciences and Mathematics Faculty,\\ College of Letters and Sciences,\\ Arizona State University, Mesa, AZ 85212. \\Email: yun.kang@asu.edu}
\address   {SAL Mathematical, Computational and Modeling Science Center, \\Arizona State University, Tempe, AZ 85287.\\ E-mail:  ccchavez@asu.edu}
\address   {Department of Agricultural, \\Food and Resource Economics, \\Michigan State University, East Lansing, MI 48824.\\ E-mail: horan@msu.edu}
\address   {School of Life Sciences, Arizona State University, Tempe, AZ 85287.\\E-mail: Charles.Perrings@asu.edu}
{\center\begin{abstract} \\
In this paper, we develop a multi-group epidemic framework via virtual dispersal where the risk of infection is  a function of the residence time and local environmental risk. This novel approach eliminates the need to define and measure contact rates that are used in the traditional multi-group epidemic models with heterogeneous mixing. We apply this approach to a general $n$-patch SIS model whose basic reproduction number $\mathcal R_0 $ is computed as a function of a patch residence-times matrix $\mathbb P$. Our analysis implies that the resulting $n$-patch SIS model has robust dynamics when patches are strongly connected: there is a globally stable endemic equilibrium when $\mathcal R_0>1 $ while the disease free equilibrium is globally stable when $\mathcal R_0\leq1 $. Our further analysis indicates that the dispersal behavior described by the residence-times matrix $\mathbb P$ has profound effects on the disease dynamics at the single patch level with consequences that proper dispersal behavior along with the local environmental risk can either promote or eliminate the endemic in particular patches.  Our work highlights the impact of residence times matrix if the patches are not strongly connected. Our framework can be generalized in other endemic and disease outbreak models. As an illustration, we apply our framework to a two-patch SIR single outbreak epidemic model where the process of disease invasion is connected to the final epidemic size relationship. We also explore the impact of disease prevalence driven decision using a phenomenological modeling approach in order to contrast the role of constant versus state dependent $\mathbb P$ on disease dynamics.
\end{abstract}}



\end{frontmatter}

\section{Introduction}\label{intro}

Sir Ronald Ross must be considered the founder of mathematical epidemiology \cite{Ross1911} despite the fact that Daniel Bernouilli (1700-1782), was most likely the first researcher to introduce the use of mathematical models in the study of epidemic outbreaks \cite{Bernoulli1766,1019.92028} nearly 150 years earlier.  Ross' appendix to his 1911 paper \cite{Ross1911} not only introduces a nonlinear system of differential equations aimed at capturing the overall dynamics of malaria contagion, a disease driven by the interactions of hosts, vectors and the life-history of \textit{Plasmodium falciparum}, but also includes a tribute to  mathematics through his observation that this framework, his model, may also be used to model the dynamics of sexually transmitted diseases  \cite{Ross1911}. Ross' observation has motivated the use of mathematics in the study of the impact of human social interaction on disease dynamics \cite{BlytheCCC89,CCCBusen91,MR1938908,MR2000i:92036,MR96k:92016,HadCasti95,HethYor84,HsuSchmitz2000b,HsuSchmitz2000,HsuSchmitz2007,Yorke:1978xu}.\\

In fact, Ross' work introduced the type of frameworks needed to capture and modify the dynamics of epidemic outbreaks;  new landscapes where public policies could be tried and tested without harming anybody, complementing and expanding the role that statistics plays in epidemiology.  Suddenly scientists and public health experts had a ``laboratory"  for assessing the impact of transmission mechanisms; evaluating, a priori, efforts aimed at mitigating or eliminating the deleterious impact of disease dynamics. \\

The study of the dynamics of communicable disease in metapopulation, multi-group or age-structure  models has also benefitted from the work of Ross. Contact matrices have been used in the study of  disease dynamics to accommodate or capture the dynamics of  heterogeneous mixing populations \cite{Anderson:1982xw,CCCCookeHuangLevin89,MR87h:92054,HethcoteSIAM2000}. The spread of communicable diseases like measles, chicken pox or rubella is  intimately connected to the the concept of contact, ``effective" contact or ``effective" \textit{per capita} contact rate \cite{CCCVelasco1994,HethcoteSIAM2000}; a clear measurable concept in, for example, the context of sexually transmitted diseases (STDs) or vector-borne diseases. The values used to define a contact matrix emerge from the \textit {a priori} belief that contacts can be clearly defined and measured in any context. Their use in the context of communicable diseases is based often on relative rankings; the result of observational subjective measures of contact or activity levels. For example, since  children are believed to have the most contacts per unit of time, their observed activity levels are routinely used to set a relative contact or activity scale. Traditionally, since school children are assumed to be the most active, they are used to set the scale with the rest of the age-specific contact matrix usually completed under the assumption of proportionate (weighted random) mixing (albeit other forms of mixing are possible \cite{AndMay91,BlytheCCC89, CCCBusen91, CCCHethcoteAndreasenLevin89, MR87h:92054,HethcoteSIAM2000} and references therein).  In short, mixing or contact matrices are used to collect re-scaled estimated levels of activity among interacting subgroups or age-classes; a phenomenological estimation process based on observational studies, surveys, and various appealing definitions of contact \cite{MossongETal2008}. 
Our belief that contact rates cannot, in general, be measured in satisfactory ways for diseases like influenza, measles or tuberculosis, arises from the difficulty of assessing the average number of contacts per unit of time of children in a school bus, or the average number of contacts per unit of time that children and adults have with each other in a classroom or at the library, per unit of time. In  some cities in Latin America, some individuals spend 2-4 hours per day as users of mass transportation systems, some traveling in packed subway cars or as ``sardines'' in small buses. How packed these modes of transportation are as a function of time of the day or day of the week can be observed but has not been uniformly quantified in terms of contacts (or age-specific contacts) per unit of time by different observers. The issue is further confounded by our inability to assess what an effective contact is: a definition that may have to be tied in to the density of floating virus particles, air circulation patterns, or whether or not  contaminated surfaces are touched by susceptible individuals. In short, defining and measuring a contact or an effective contact, turns out to be incredibly challenging \cite{MossongETal2008}. That said, experimental methods may be used to estimate the average risk of acquiring, for example, tuberculosis (TB) or influenza, to individuals that spend on the average 3 hours per day in public transportation, in Mexico City or New York City.\\

In this paper we propose the use of residence times in heterogeneous environments, as a proxy for ``effective" contacts over an ``$x$" windows in time. Catching a communicable disease would of course depend on the presence of infected/infectious individuals (a necessary condition), the level of ``risk" within a given ``patch'' (crowded bars, airports, schools, work places, etc), and the time spent in such environment. Risk of infection is assumed to be a function  of the time spent in pre-specified environments;  risk that may be experimentally measured. We argue that characterizing a landscape as a collection of patches defined by risk (public transportation, schools, malls, work place, homes, etc) is possible, especially if the risk of infection in such ``local'' environments is in addition a function of residence times and disease levels. Ranking patch-dependent  risks of infection via the values of the transmission rate ($\beta$) per unit of time, may therefore be possible and useful. The reinterpretation of $\beta$ and the use of residence times move us away from the world of models that account for transmission via the use of differential susceptibility to the world where infection depends on local environmental risk.\\

Consequently, we introduce a residence times framework in the context of a multi-group system defined by patch-dependent risk (defined by $\beta$). We study the role of patch residence times on disease dynamics within endemic and single outbreak multi-group scenarios. Specifically, the study of the impact of patch residence times (modeled by a matrix of constants) on disease dynamics within a Susceptible-Infected-Susceptible ($SIS$) framework is carried out first, under the philosophy found in \cite{Brauer2008,BrauerCCC94,BrauerCCCVelasH96,BraVdd01,BrauerWatmough2009,MR1938908,HadCasti95,HuangCookeCCC92}. Individuals move across patches as a function of their assessment of relative levels of infection in each area (studies using alternative classical approaches are found in  \cite{BraVdd01,BrauerVddWang2008, HeiderichHuangCCC2002,ValascoHernandezBrauerCCC96}). Generalizations are explored through simulations of the two-patch SIS model with state-dependent residence times within our framework. The results are compared to the disease dynamics generated by constant residence times. \\

The rest of this paper is organized as follows: Section \ref{sec:1GeneralModel} introduces a general $n$ patch $SIS$ model that accounts for residence times. Theoretical results on the role of residence times matrix ($\mathbb P$) on disease dynamics are carried out using the residence times dependent basic reproduction number $\mathcal R_0(\mathbb P)$. Patch-specific reproduction numbers $\mathcal R_0^i(\mathbb P)$, $i=1,\dots,n$ are used to highlight the impact of residence times matrices on cases that includes non-strongly connected patch configuration. Section \ref{StateDependant2Patches} explores, through simulations, the dynamics of the $SIS$ model under a state-dependent residence times matrix in a two-patch system;  $\mathbb P\equiv\mathbb P(I_1,I_2)$. That is, when the decisions to spend time in a patch are a function of patch-disease prevalence. Section \ref{sec:FES} highlights our framework in the case of a two patch single outbreak $SIR$ model following the work of Brauer \cite{Brauer2008,BrauerWatmough2009}, and discusses the role of $\mathbb P$ on the final epidemic size. Section \ref{CCLDiscussions} collects our observations, conclusions and discusses future work. The detailed proofs of our theoretical results are provided in the Appendix.\\

\section{A general $n$-patch $SIS$ model with residence times}
\label{sec:1GeneralModel}
A general $n$-patch SIS model with residence time matrix $\mathbb P$ is derived. The global analysis of the model is carried out via the basic reproduction number $\mathcal R_0$. We also include patch-dependent disease persistence conditions. \\
\subsection{Model derivation}

We model disease dynamics within an environment defined by $n$ patches (or risk areas) and so, we let $N_i(t), i=1,2...,n$ denote resident population at Patch $i$ at time $t$. We assume that Patch $i$ residents spend $p_{ij}\in [0,1]$ time in Patch $j$, with $\sum_{j=1}^{j=n} p_{ij}=1$, for each $i=1,\dots,n$. 
 In extreme cases, for examples, we may have, for $p_{ij}=0$, $i\neq j$, that is Patch $i$ residents spend no time in Patch $j$ while $\sum_{j\neq i}p_{ij}=1$  (or equivalently $p_{ii}=0$) would imply that Patch $i$ residents spend all their time in Patch $j$ (with $j=1,\dots, n$ and $j\neq i$) even though their patch is (labelled) $i$. In the absence of disease dynamics, the population of Patch $i$ residents is modeled by the following equation:
\bae\label{ni}
\frac{dN_i}{dt}&=&b_i-d_i N_i
\eae where the parameters $b_i$, $d_i$ represent the birth rate, and the natural \textit{per capita} death rate in Patch $i$, respectively. Hence, the Patch $i$  resident population approaches the constant $\frac{b_i}{d_i}$ as $t\to\infty$. \\

In the presence of disease, we assume that disease dynamics are captured by an $SIS$ model, thus,  the Patch $i$  resident population is divided into susceptible and infected classes, represented by  $S_i, \, I_i$, respectively, with $S_i+I_i=N_i$.  We further assume that (a) there is no additional death due to disease; (b) the Patch $i$ Infected resident population recovers and goes back to the susceptible class at the \textit{per capita} rate $\gamma_i$; (c) the residence time matrix $\mathbb P=(p_{ij})_{i=1,..,n}^{j=1,..,n}$ collects the proportion of times spent by $i$-residents in $j$-environments, $i=1,\dots, n$ and $j=1, \dots, n$. The disease dynamics are therefore described by the following equations:
\bae\label{nM}\begin{array}{lcl}
\dot S_i&=&b_i-d_iS_i+\gamma_i I_i-\sum_{j=1}^n(\textrm{$S_i$ infected in Patch } j)\\
\dot I_i&=&\sum_{j=1}^n(\textrm{$S_i$ infected in Patch } j)- \gamma_i I_i-d_i I_i\\
\dot N_i&=&b_i-d_iN_i.
\end{array}\eae
We model $S_i$ infection within Patch $j$ in the following way: 
\begin{itemize}
\item Since each $p_{ij}$ entry of  $\mathbb P$ denotes the \textit{proportion of time} that Patch $i$ residents spent mingling in Patch $j$, we have that:
\begin{itemize}
\item There are  $p_{ij} N_i=p_{ij} S_i + p_{ij}I_i$ Patch $i$ residents in Patch $j$ on the average at time $t$.
\item The total Patch $j$, the total effective population is $\sum_{k=1}^{n}p_{kj}N_k$, of which $\sum_{k=1}^{n}p_{kj}I_k$ are infected. Hence, the proportion of infected individuals in Patch $j$ is $\frac{\sum_{k=1}^{n}p_{kj}I_k}{\sum_{k=1}^{n}p_{kj}N_k}$ and well defined, as long as there exists a $k$ such that $p_{kj}>0$; so that the population in Patch $j$ is nonzero.
\end{itemize}
\item Hence, the $[\textrm{$S_i$ infected per unit of time in Patch } j]$ can be represented as the product of the following three items:
$$\underbrace{\beta_j}_{\textbf{the risk of infection in Patch $j$ }}\times \underbrace{p_{ij}S_i}_{\textbf{Susceptible from Patch $i$ who are currently in Patch $j$ }}$$ $$\times \underbrace{\frac{\sum_{k=1}^{n}p_{kj}I_k}{\sum_{k=1}^{n}p_{kj}N_k}}_{\textbf{Proportion of infected in Patch $j$}}.$$

The transmission takes on a modified frequency-dependent form that depends on how much time individuals of each epidemiological class spend in a particular area, and where $\beta_j$ differs by patch to reflect spatial differences in potential infectivity. More precisely, $\beta_j$ is assumed to be a patch-specific measure of disease risk per unit of time with its effectiveness tied in to local environmental and sanitary conditions. Therefore,
\bae\label{SifromJ}
[\textrm{$S_i$ infected per unit of time in Patch } j]\equiv\beta_j\times p_{ij}S_i\times \frac{\sum_{k=1}^{n}p_{kj}I_k}{\sum_{k=1}^{n}p_{kj}N_k}\eae provided that there exists $k$ such that $p_{kj}>0$.

\end{itemize}
Model \eqref{nM} can be rewritten as follows:
\bae\label{nMc}\begin{array}{lcl}
\dot S_i&=&b_i-d_iS_i+\gamma_i I_i-\sum_{j=1}^n\left(\beta_j\times p_{ij}S_i\times \frac{\sum_{k=1}^{n}p_{kj}I_k}{\sum_{k=1}^{n}p_{kj}N_k}\right),\\
\dot I_i&=&\sum_{j=1}^n\left(\beta_j\times p_{ij}S_i\times \frac{\sum_{k=1}^{n}p_{kj}I_k}{\sum_{k=1}^{n}p_{kj}N_k}\right)- \gamma_i I_i-d_i I_i,\\
\dot N_i&=&b_i-d_iN_i,
\end{array}\eae
with, the dynamics of the Patch $i$ resident total population modeled  by the equation: $\dot N_i(t)=b_i-d_iN_i(t)$, where $S_i+I_i=N_i$, which implies that $\displaystyle N_i(t)\to\frac{b_i}{d_i}$ as $t\to+\infty$. Theory of asymptotically autonomous systems for triangular systems \cite{CasThi95,0478.93044} guaranties that System (\ref{nMc}) is asymptotically equivalent to:

 \begin{equation} \label{IGenlDetail}
 \begin{array}{lcl}
\dot I_i&=&\sum_{j=1}^n\left(\beta_j p_{ij}\left(\frac{b_i}{d_i}-I_i\right) \frac{\sum_{k=1}^{n}p_{kj}I_k}{\sum_{k=1}^{n}p_{kj}\frac{b_k}{d_k}}\right)- (\gamma_i +d_i) I_i\\
&=&I_i\left(\frac{b_i}{d_i}-I_i\right)\left(  \sum_{j=1}^{n}\frac{\beta_jp_{ij}^2}{\sum_{k=1}^{n}p_{kj}\frac{b_k}{d_k}}\right)+\left(\frac{b_i}{d_i}-I_i\right)\sum_{j=1}^{n}\frac{\beta_jp_{ij}\sum_{k=1,k\neq i}^{n}p_{kj}I_k}{\sum_{k=1}^{n}p_{kj}\frac{b_k}{d_k}}  -(d_i+\gamma_i )I_i  
\end{array}\end{equation}
for $i=1,2,\dots,n$, with residence times matrix $\mathbb P=\left(p_{ij}\right)_{i=1,...,n}^{j=1,..,n}$ satisfying the conditions: \\

\noindent\textbf{HP1.} At least one entry in each column of $\mathbb P$ is strictly positive; and \\

\noindent\textbf{HP2.} The sum of all entries in each row is one, i.e., $\sum_{j=1}^n p_{ij}=1$ for all $i$.\\

%
%
\subsection{Equilibria, basic reproduction number and global analysis }

To analyze the system, we investigate the basic reproduction number of the system with fixed residence times to better understand its properties in the absence of behavioral responses to risk. We let $\mathcal B=\left(\beta_1,\beta_2,\dots,\beta_n\right)^t$ define the risk of infection vector;  $\beta_i$ a measure of the risk per susceptible per unit of time while in residence in Patch $i$. \\

Letting  
$\displaystyle S=
\left(S_1, S_2, \dots, S_n\right)^t,\;I=\left(I_1, I_2, \dots,I_n\right)^t,\;$ $\displaystyle \bar N=\left(
\frac{b_1}{d_1},\frac{b_2}{d_2},\dots, \frac{b_n}{d_n}\right)^t$, and $\tilde N =\mathbb P^t\bar N= \left(\begin{array}{cc}
\sum_{k=1}^{n}p_{k1}\frac{b_k}{d_k}  \\
\sum_{k=1}^{n}p_{k2}\frac{b_k}{d_k} \\
\vdots\\
\sum_{k=1}^{n}p_{kn}\frac{b_k}{d_k} 
\end{array}\right).$
Then System (\ref{IGenlDetail}) can be rewritten in the following compact (vectorial) form:

 \begin{equation} \label{SIScompact}
\dot I=\textrm{diag}(\bar N-I)\mathbb P\textrm{diag}(\mathcal B)\textrm{diag}(\tilde N)^{-1}\mathbb P^tI-\textrm{diag}(d_I+\gamma_I)I
\end{equation} with state space in $\mathbb R^n_+$. System \eqref{SIScompact} has the compact set $\Omega=\{I\geq\mathbf{0}_{\mathbb R^n},\;I\leq\bar N \}$ as its global attractor.  This implies that the populations involved are ``biologically" well-defined since solutions of \eqref{SIScompact} will converge to and stay in $\Omega$. We therefore restrict the dynamics of \eqref{SIScompact} to the compact set $\Omega$.\\%

  The analysis of System \eqref{SIScompact} is naturally tied in to the basic reproductive number $\mathcal R_0$ \cite{MR1057044,VddWat02}; the average number of secondary cases produced by an infected individual during its infectious period while interacting with a purely susceptible population.  $\mathcal R_0$ is given by (see the detailed formulation in Appendix):
\bae\label{cR0}
\mathcal R_0&=&\rho(-\textrm{diag}(\bar N)\mathbb P\textrm{diag}(\mathcal B)\textrm{diag}(\tilde N)^{-1}\mathbb P^tV^{-1})\eae
where $V=-\textrm{diag}(d_I+\gamma_I)$, $\displaystyle d_I=\left(d_1, d_2, \dots,d_n\right)^t$ and $\displaystyle \gamma_I=\left(\gamma_1, \gamma_2, \dots,\gamma_n\right)^t$.\\

 The basic reproduction number $\mathcal R_0$ is used to establish global properties of System (\ref{SIScompact}). 
For the relevant literature on global stability for multi-group or metapopulation models, see \cite{ArVddR003,IggidrSalletTsanou,KuniyaMuroya2014,LajYo76,SattenspielSimon88} and the references therein. We  define the disease free equilibrium (DFE) of System \eqref{SIScompact} as $I^\ast=\mathbf{0}_{\mathbb R^n}$ and the endemic equilibrium  (when $\mathcal R_0>1$) as $\bar I$ where all components are positive. By using the same approach as in \cite{IggidrSalletTsanou,LajYo76}, we arrive at the following theorem regarding the global dynamics of  Model (\ref{SIScompact}).

\begin{theorem}\label{MainTheo}[Global dynamics of Model (\ref{SIScompact})] Suppose that the residence times matrix  $\mathbb P$ is irreducible, then the following statements hold:
\begin{itemize}
\item If $\mathcal R_0\leq1$, the DFE $I^\ast=\mathbf{0}_{\mathbb R^n}$ is globally asymptotically stable. If $\mathcal R_0>1$ the DFE is unstable.
\item If $\mathcal R_0>1$, there exists a unique endemic equilibrium $\bar I$ which is GAS.
\end{itemize}
\end{theorem}

\noindent\textbf{Remarks:} The detailed proof of Theorem \ref{MainTheo} is provided in Appendix \ref{AppProof}.  These results imply that System \eqref{SIScompact} is robust, that is, disease outcomes are completely determined by whether or not the reproduction number $\mathcal R_0$ is greater or less than one. The results of Theorem \ref{MainTheo} while powerful, do not provide easily accessible insights on the impact of the residence matrix $\mathbb P$ on the levels of infection within each patch. \\

Direct insights on the effects of $\mathbb P$, are derived by focusing on the levels of endemicity within each patch. The following two definitions help set the stage for the discussion:

\begin{itemize}
\item The basic reproduction number for Patch $i$ in the absence of movement ($p_{ii}=1$ or $\sum_{i\neq j}p_{ij}=0$), $SIS$ model, is defined as $\mathcal R_0^i\equiv\frac{\beta_i}{d_i+\gamma_i}$, which determines whether or not the disease will be endemic in Patch $i$. In short disease will die out if $\mathcal R_0^i\leq 1$ with a unique endemic equilibrium, that is GAS, if $\mathcal R_0^i>1$.
\item The basic reproduction number  associated with Patch $i$, under the presence of multi-patch residents, is defined as follows:
 $$ \begin{array}{lcl}
 \mathcal R_0^i(\mathbb P)&=&\frac{\sum_{j=1}^n\beta_j \left(p_{ij} \frac{b_i}{d_i}\right)\left( \frac{p_{ij}}{\sum_{k=1}^{n}p_{kj}\frac{b_k}{d_k}}\right)}{d_i+\gamma_i}=\frac{\sum_{j=1}^n\beta_j p_{ij}\left( \frac{\left(p_{ij} \frac{b_i}{d_i}\right)}{\sum_{k=1}^{n}p_{kj}\frac{b_k}{d_k}}\right)}{d_i+\gamma_i}\\
 &=&{\mathcal R}_0^i \times \sum_{j=1}^n\left(\frac{\beta_j}{\beta_i}\right) p_{ij}\left( \frac{\left(p_{ij} \frac{b_i}{d_i}\right)}{\sum_{k=1}^{n}p_{kj}\frac{b_k}{d_k}}\right).
 \end{array}
 $$
\end{itemize}
We explore the role  that $\mathcal R_0^i(\mathbb P)$ plays in determining the impact of all residents on disease dynamics persistence in  Patch $i$ in the following theorem.
\begin{theorem}\label{MainTheo2}[The endemicity of disease in Patch $i$] Assume that the residence times matrix  $\mathbb P$ satisfies Condition \textbf{HP1} and \textbf{HP2} but that some of its entries can be zeros.
\begin{itemize}
 \item If $\mathcal R_0^i(\mathbb P)>1$, then the disease persists in Patch $i$. 
 \item If the following conditions hold:
$$\mbox{ \textbf{H: }}p_{kj}=0 \mbox{ for all  }k=1,..,n, \mbox{ and } k\neq i, \,\,\mbox{ whenever } p_{ij}>0,$$ then we have
$$\mathcal R_0^i(\mathbb P)={\mathcal R}_0^i \times \sum_{j=1}^n\left(\frac{\beta_j}{\beta_i}\right) p_{ij}\left( \frac{\left(p_{ij} \frac{b_i}{d_i}\right)}{\sum_{k=1}^{n}p_{kj}\frac{b_k}{d_k}}\right)={\mathcal R}_0^i \times \sum_{j=1}^n \left(\frac{\beta_j}{\beta_i}\right) p_{ij}.$$Thus, when Condition \textbf{H} holds and ${\mathcal R}_0^i \times \sum_{j=1}^n \left(\frac{\beta_j}{\beta_i}\right) p_{ij}<1$, then
endemic levels of disease cannot be supported in Patch $i$. That is, $$\lim_{t\rightarrow\infty} I_i(t)=0.$$
\end{itemize}
\end{theorem}
 \noindent\textbf{Remarks:} The detailed proof of Theorem \ref{MainTheo2} is provided in Appendix \ref{AppProof2}. The results of Theorem \ref{MainTheo2} give insights on the role that the infection risk (measured by $\mathcal B$) and the residence time matrix ($\mathbb P$) have in  promoting or suppressing infection. Further, a closer look at  the expression of the general basic reproduction number in Patch $i$, namely
  $$\mathcal R_0^i(\mathbb P)={\mathcal R}_0^i \times \sum_{j=1}^n\left(\frac{\beta_j}{\beta_i}\right) p_{ij}\left( \frac{\left(p_{ij} \frac{b_i}{d_i}\right)}{\sum_{k=1}^{n}p_{kj}\frac{b_k}{d_k}}\right),$$
 leads to the following observations:
 \begin{enumerate}
 \item The movement between patches, modeled via residence time matrix $\mathbb P$, can promote endemicity: For example, if ${\mathcal R}_0^i =\frac{\beta_i}{d_i+\gamma_i}\leq 1$, i.e., there is no endemic disease in patch $i$. Then, the presence of movement connecting Patch $i$ to possibly all other patches can support endemic disease levels in the following ways:
 
 \begin{itemize}
 \item Via the presence of high risk patches, that is, there exists a patch $j$ such that $\frac{\beta_j}{\beta_i}$ is large enough. For example, letting $p_{kl}=1/n$ for all $k, l$ with the total population in each patch being the same ($\frac{b_k}{d_k}=K$ for all $k$; $K$ a constant) then $\mathcal R_0^i(\mathbb P)={\mathcal R}_0^i \frac{\sum_{j=1}^n \beta_j}{n\beta_i}$ and consequently, if $\sum_{j=1}^n \beta_j>\frac{n \beta_i}{\mathcal R_0^i}$, then Patch $i$ will  promote the disease at endemic levels.
 \item Whenever individuals spend more time in high risk than in low risk patches. For example, in the extreme case, $p_{ij}=1$ with $\frac{\beta_j}{\beta_i}>\frac{1}{\mathcal R_0^i}$, we have that $\mathcal R_0^i(\mathbb P)>1$, and  thus, endemic disease levels in Patch $i$ can be supported. Patch $j$ ($j=1, \dots, n$ and $j\neq i$) can therefore be considered  the source and Patch $i$ ($i\neq j$) the sink \cite{Arino08,ArDavHarJorVdd05,ArVddR003,ArVdd06,KuniyaMuroya2014,LajYo76,SattenspielSimon88,7606146}.
 \end{itemize}
 \item Under the assumption $\mathcal R_0^i >1$, for an isolated Patch $i$, conditions that lead to disease extinction in the same Patch $i$ under the movement  can be identified. According to Theorem \ref{MainTheo2}, Condition \textbf{J} should be satisfied and so the expression of $\mathcal R_0^i(\mathbb P)$ reduces to
 $$\mathcal R_0^i(\mathbb P)={\mathcal R}_0^i \times \sum_{j=1}^n\left(\frac{\beta_j}{\beta_i}\right) p_{ij}\left( \frac{\left(p_{ij} \frac{b_i}{d_i}\right)}{\sum_{k=1}^{n}p_{kj}\frac{b_k}{d_k}}\right)={\mathcal R}_0^i \times \sum_{j=1}^n \left(\frac{\beta_j}{\beta_i}\right) p_{ij}.$$
Therefore, the only way to have  the value of $\mathcal R_0^i(\mathbb P)$ be less than one, would be when the amount of time spent in Patch $i$ is such that $\sum_{j=1}^n \left(\frac{\beta_j}{\beta_i}\right) p_{ij}<\frac{1}{\mathcal R_0^i(\mathbb P)}$. Therefore, we conclude that the synergy between the residence time matrix $\mathbb P$ and the existence of sufficient low risk patches (i.e.,  $\beta_j\ll\beta_i$) can  suppress a disease outbreak in Patch $i$.
 
 \end{enumerate}

\section{Two patch models: state-dependent residence times matrix}\label{StateDependant2Patches}

We now extend the analysis of disease dynamics to the case where susceptible individuals respond to variations in risk in an automatic way. In particular, we consider the case when susceptible individuals make programmed responses to variations in disease risk, and do not choose their response to optimize an index of wellbeing ( see for example \cite{Brauer2008,BrauerCCCVelasH96,BraVdd01,BrauerWatmough2009}).  While this may not be a very good approximation of disease risk management in real systems, it enables us to explore the implications of certain types of phenomenologically modeled behavioral responses by assuming, for example, that the proportion of time spent in a particular patch depends on the numbers of infected individuals on that particular patch; that is  $\mathbb P\equiv\mathbb P(I_1,I_2)$.\\
 
Possible properties of the proportion of time spent by resident of Patch $i$ into Patch $j$, $i\neq j$, ($p_{ij}$) may include: increases with respect to the growth of infected resident in patch $i$ ($I_i$), or decreases with respect to infected resident in patch $j$ ($I_j$). Mathematically, we would have that
$$\frac{\partial p_{ij}(I_i,I_j)}{\partial I_j}\leq0\quad\textrm{and}\quad \frac{\partial p_{ij}(I_i,I_j)}{\partial I_i}\geq0.$$

In a two-patch system, the use of the relationship $p_{ij}(I_1,I_2)+p_{ji}(I_1,I_2)=1$, reduces the above four conditions on $\mathbb P$, to the following conditions:

$$\frac{\partial p_{11}(I_1,I_2)}{\partial I_1}\leq0\quad\textrm{and}\quad \frac{\partial p_{22}(I_1,I_2)}{\partial I_2}\leq0.$$

Examples of functions $p_{ij}(I_1,I_2)$ with these properties include,

$$p_{12}(I_1,I_2)=\sigma_{12}\frac{1+I_1}{1+I_1+I_2}\quad\textrm{and}\quad p_{21}(I_1,I_2)=\sigma_{21}\frac{1+I_2}{1+I_1+I_2}$$

and $$p_{11}(I_1,I_2)=\frac{\sigma_{11}+\sigma_{11}I_1+I_2}{1+I_1+I_2}\quad\textrm{and}\quad p_{22}(I_1,I_2)=\frac{\sigma_{22}+I_1+\sigma_{22}I_2}{1+I_1+I_2}$$
where $\sigma_{ij}$ are such that $\displaystyle\sum_{j=1}^{2}\sigma_{ij}=1$. \\

More complex behavioral response formulations may also depend on the states of total populations $N_1$ and $N_2$, but the current specification captures important components of risk (infections) and allows us to retain the asymptotic equivalence property applied in the case of fixed residence times. Hence, using the same notation as in System (\ref{SIScompact}) leads to the following two dimensional system with $\mathbb P=\mathbb P(I_1,I_2)$:

\begin{equation} \label{2PNI3Variable}
\left\{\begin{array}{ll}
\dot I_1=X(I_1,I_2)(\frac{b_1}{d_1}-I_1)I_1    +  Y(I_1,I_2)(\frac{b_1}{d_1}-I_1)I_2              -(d_1+\gamma_1 )I_1,\\\\
\dot I_2=Y(I_1,I_2)(\frac{b_2}{d_2}-I_2)I_1    + Z(I_1,I_2)(\frac{b_2}{d_2}-I_2)I_2   -(d_2+\gamma_2 )I_2 ,  
\end{array}\right.
\end{equation}where $$X(I_1,I_2)=\frac{\beta_1p^2_{11}(I_1,I_2)}{p_{11}(I_1,I_2)\frac{b_1}{d_1}+p_{21}(I_1,I_2)\frac{b_2}{d_2}}+\frac{\beta_2p^2_{12}(I_1,I_2)}{p_{12}(I_1,I_2)\frac{b_1}{d_1}+p_{22}(I_1,I_2)\frac{b_2}{d_2}},$$

$$Y(I_1,I_2)=\frac{\beta_1p_{11}(I_1,I_2)p_{21}(I_1,I_2)}{p_{11}(I_1,I_2)\frac{b_1}{d_1}+p_{21}(I_1,I_2)\frac{b_2}{d_2}}+\frac{\beta_2p_{12}(I_1,I_2)p_{22}(I_1,I_2)}{p_{12}(I_1,I_2)\frac{b_1}{d_1}+p_{22}(I_1,I_2)\frac{b_2}{d_2}},$$
and 
$$Z(I_1,I_2)=\frac{\beta_1p^2_{21}(I_1,I_2)}{p_{11}(I_1,I_2)\frac{b_1}{d_1}+p_{21}(I_1,I_2)\frac{b_2}{d_2}}+\frac{\beta_2p^2_{22}(I_1,I_2)}{p_{12}(I_1,I_2)\frac{b_1}{d_1}+p_{22}(I_1,I_2)\frac{b_2}{d_2}},$$where $X(I_1,I_2)$, $Y(I_1,I_2)$ and $Z(I_1,I_2)$ are positive functions of $I_1$ and $I_2$. \\

The basic reproduction number $\mathcal R_0$ is the same as in the previous section since it is computed at the infection-free state, i.e.
$$\mathcal R_0=\rho(\text{diag}(\bar N)\mathbb P\text{diag}(\mathcal B)\text{diag}(\tilde N)^{-1}\mathbb P^t(-V^{-1}))$$where, in this case, we have that
$\mathbb P=\begin{bmatrix}
\sigma_{11} & \sigma_{12}\\\\
\sigma_{21}&\sigma_{22}
  \end{bmatrix}\quad \textrm{and}\quad \sigma_{ij}=p_{ij}(0,0),\;\;  \forall \{i, j\}=\{1,2\}.$\\

The properties of positiveness and boundedness of trajectories of System (\ref{SIScompact}) are preserved in System  (\ref{2PNI3Variable}). In addition, System (\ref{2PNI3Variable}) has  a unique DFE equilibrium whose local stability is determined by the value of the (uncontrolled) $\mathcal R_0$: the DFE is locally asymptotically stable if $\mathcal R_0<1$ while it is unstable if  $\mathcal R_0>1$.\\ 

Let us consider whether System (\ref{2PNI3Variable}) can have a boundary equilibrium such as $(0,\bar I_2)$ or $(\bar I_1,0)$. The assumption that System (\ref{2PNI3Variable})  has such a boundary equilibrium $(0,\bar I_2)$ with $\bar I_2>0$ implies that $Y(0,\bar I_2)=0$. 
 Since $p_{11}(0,I_2)=\frac{\sigma_{11}+I_2}{1+I_2}$ and $p_{22}(0,I_2)=\sigma_{22}$, we deduce that
$$Y(0,I_2)=\frac{\beta_1\sigma_{21}(\sigma_{11}+I_2)}{\frac{\sigma_{11}+I_2}{1+I_2}\frac{b_1}{d_1}+\sigma_{21}\frac{b_2}{d_2}(1+I_2)}+\frac{\beta_2\sigma_{12}\sigma_{22}}{\sigma_{12}\frac{b_1}{d_1}+\sigma_{22}\frac{b_2}{d_2}(1+I_2)}.
$$
This indicates that $Y(0,I_2)=0$ if and only if $\sigma_{21}=0$ and $\sigma_{12}=0$, which requires that:
$$p_{12}=p_{21}=0,\,\,\mbox{ and }p_{11}=p_{22}=1.$$
 A similar arguments can be applied to the boundary equilibrium $(\bar I_1,0)$. Therefore, we conclude that System (\ref{2PNI3Variable}) will have  a boundary equilibrium $\left(\frac{}{}(0,\bar I_2)\;\; \textrm{or}\;\; (\bar I_1,0)\right)$ only in the trivial case of isolated patches, that is, where there is no movement between two patches. This conclusion differs from the state-independent residence matrix model \eqref{SIScompact}, since for example, the two-patch model \eqref{SIScompact}, according to Theorem \ref{MainTheo}, boundary equilibrium $(0,\bar I_2)$ or $(\bar I_1,0)$ can exist when $p_{11}=p_{22}=0$ ($p_{12}=p_{21}=1$).\\

To illustrate the difference between the state-dependent residence matrix model  (\ref{2PNI3Variable}) and the state-independent residence matrix model \eqref{SIScompact}, we look at the situation when $\sigma_{11}=\sigma_{22}=0, \sigma_{12}=\sigma_{21}=1$ ( $p_{11}=p_{22}=0, p_{12}=p_{21}=0$ for the state-independent residence matrix model \eqref{SIScompact}). Under the condition of $\sigma_{11}=\sigma_{22}=0, \sigma_{12}=\sigma_{21}=1$, we have Model  (\ref{2PNI3Variable}), that

 $$p_{12}(I_1,I_2)=\frac{1+I_1}{1+I_1+I_2}\quad\textrm{and}\quad p_{21}(I_1,I_2)=\frac{1+I_2}{1+I_1+I_2}$$

and $$p_{11}(I_1,I_2)=\frac{I_2}{1+I_1+I_2}\quad\textrm{and}\quad p_{22}(I_1,I_2)=\frac{I_1}{1+I_1+I_2}.$$
This difference  has significant impact on disease dynamics (see Fig ~\ref{I1Loop} and Fig ~\ref{I2Loop}, red curves).\\

 In  Fig ~\ref{I2Loop}, we see that the infection in Patch 2 (high risk)  persists in the state-dependent case whereas it dies out when $\mathbb P$ is constant. That is due to the fact that $p_{ii}(I_1,I_2)$ will not equal zero whereas  $p_{ij}(I_1,I_2)$ with $i\neq j$ may. 
 For the constant residence times matrix, the dynamics of the disease in each patch is also independent, where people in patch $i$ infect only susceptible in patch $j$ with $i\neq j$. In Fig \ref{I2Loop} (red solid curve), we observe that the disease dies out in Patch 2 with $\tilde{\mathcal {R}}^2_0=\frac{\beta_1}{d_2+\gamma_2}=0.8571$. For the state-dependent case, unless there is no disease in both patches or one disease-free Patch, the proportion of time residents spend in their own patch is nonzero. This leads the disease to persist in both patches if $\mathcal R_0>1$ (see  Fig \ref{I2Loop}, red dashed curves). However, even in this case, the disease dies out in both patches if $\mathcal R_0<1$ (See Fig \ref{SDR0Less1}, red curves, for instance).

\begin{figure}[ht]
\centering
 \subfigure[Dynamics of the disease in Patch 1. If there is no movement between the patches (blue curves), the disease dies out in the low risk Patch 1 in both approaches with $\mathcal R^1_0=0.7636$.]{
   \includegraphics[scale =.22]{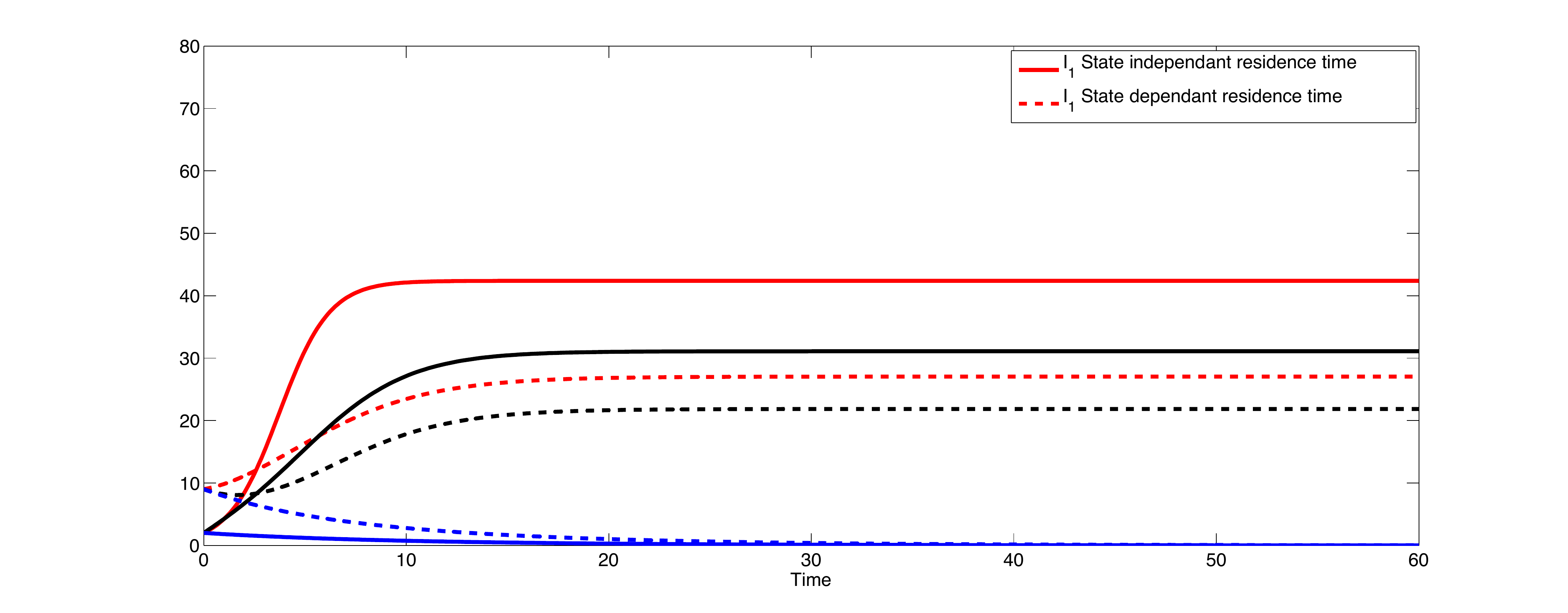}\label{I1Loop} }   
 \subfigure[Dynamics of the disease in Patch 2. In the high mobility case, the disease dies out (solid red curve) for $\mathbb P$ constant, with $\tilde{\mathcal {R}}^2_0=0.8571$, and persists for $\mathbb P$ state-dependent (dashed red curve).]{
   \includegraphics[scale =.22]{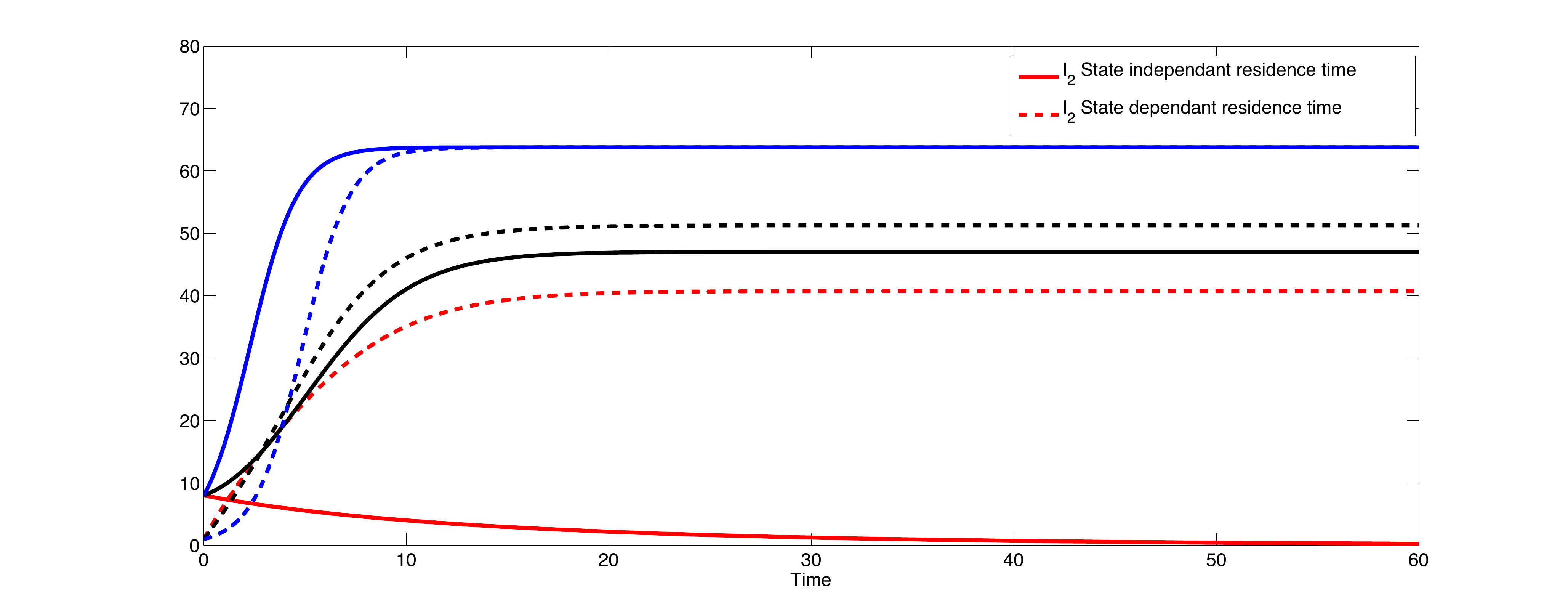}\label{I2Loop} }
  \caption{Coupled Dynamics of $I_1$ and $I_2$ for constant $p_{ij}$ (solid) and state dependent $p_{ij}$ (dashed). The red lines is case of high mobility, i.e. $p_{12}=p_{21}=\sigma_{12}=\sigma_{21}=1$. The black lines represent the symmetric case, i.e: $p_{12}=p_{21}=\sigma_{12}=\sigma_{21}=0.5$ and the blue line represent the polar case, i.e.: $p_{12}=p_{21}=\sigma_{12}=\sigma_{21}=0$.} \label{fig:twofigs}
\end{figure}

\begin{figure}[ht!]
\includegraphics[scale=0.4]{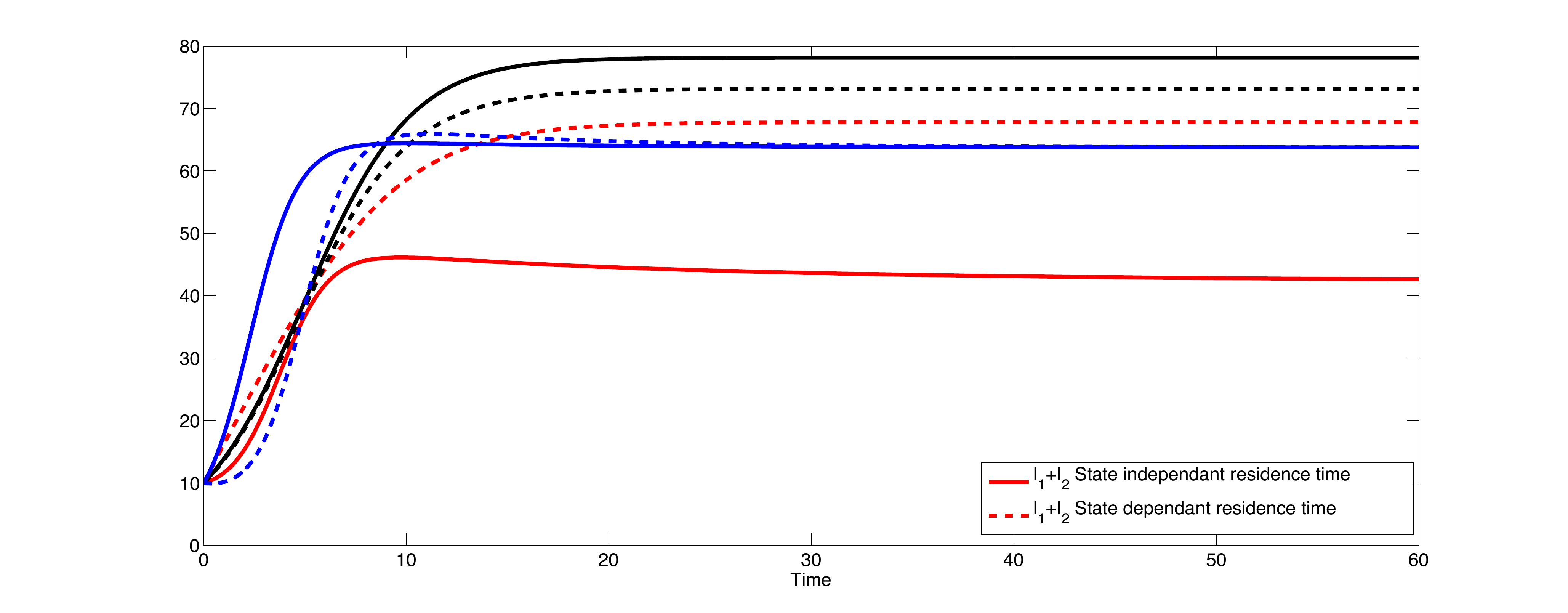}
\caption{Coupled Dynamics of $I_1+I_2$ for constant $p_{ij}$ (solid) and state dependent $p_{ij}$ (dashed). The overall prevalence is higher if if the residence times is symmetric (solid and dashed black curves). The black curves represent the symmetric case ( $p_{12}=p_{21}=\sigma_{12}=\sigma_{21}=0.5$ ), and the blue line represent the polar case ( $p_{12}=p_{21}=\sigma_{12}=\sigma_{21}=0$) and red curves represent high mobility case ($p_{12}=p_{21}=\sigma_{12}=\sigma_{21}=1)$.}
\label{PrevalenceLoop} 
\end{figure}

\subsection{Applications and comparisons:  the two patch cases} 

The analytical results of the global dynamics on the asymptotic behavior of Model  (\ref{2PNI3Variable}) are still unresolved. Hence, we ran simulations to gain some insights on the role of $\mathbb P(I_1,I_2)$ on endemic dynamics. We observe that trajectories converge towards an endemic equilibrium whenever $\mathcal R_0>1$; however, there are substantial differences in the transient dynamics generated by state-dependent $\mathbb P(I_1,I_2)$ when compared to those generated with a constant residence times matrix. \\

Unless stated otherwise, we suppose the following generic values for the simulations: $\beta_1=0.3,\; \beta_2=1.2,\; b_1=9,\;\; d_1=1/7,\; b_2=9,\; d_2=1/10$ and $\gamma_1=\gamma_2=1/4$. From a selected of simulations, it is observed that:
\begin{enumerate}
\item For the symmetric case where $p_{12}=p_{21}=0.5$, the disease is endemic in both patches as predicted by Theorem \ref{MainTheo} since $\mathcal R_0=2.0466$. For the state-dependent case, simulations suggest (Fig \ref{I1Loop} and Fig \ref{I2Loop}, black dashed curves) that trajectories tend to be endemic in both patches.  However, the level of endemicity is lower than the constant case in Patch 1 (low risk patch) and is greater in Patch 2 (high risk patch).\\
\item Fig \ref{PrevalenceLoop} sketches the overall prevalence in both patches with three different scenarios of residence times matrix $\mathbb P$, both the constant and state-dependent case.  The disease persists since the overall $\mathcal R_0>1$ in all three cases. \\
\item The case where there is no movement between patches, that is, $p_{12}=p_{21}=0$ ( $p_{11}=p_{22}=1$) and $\sigma_{12}=\sigma_{21}=0$ (or $p_{12}(I_1,I_2)=p_{21}(I_1,I_2)=0$), corresponds to the case where  the system behaves as two isolated patches.  In this case the disease dies out or persists in Patch $i$ if $\mathcal R_0^i$ is above or below unity in both approaches. This is illustrated on Fig \ref{I1Loop}  and  Fig \ref{I2Loop} where the disease dies out in Patch 1 ( Fig \ref{I1Loop}, blue solid line) where $\mathcal R_0^1=\frac{\beta_1}{d_1+\gamma_1}=0.7636$ and the disease persists in Patch 2 (Fig \ref{I2Loop}, blue solid curve) where $\mathcal R_0^2=\frac{\beta_2}{d_2+\gamma_2}=3.4286$. For the state dependent case ( dashed blue curves on on Fig \ref{I1Loop}  and  Fig \ref{I2Loop}) the outcome is similar to the constant residence times case.\\

\begin{figure}
\includegraphics[scale=0.4]{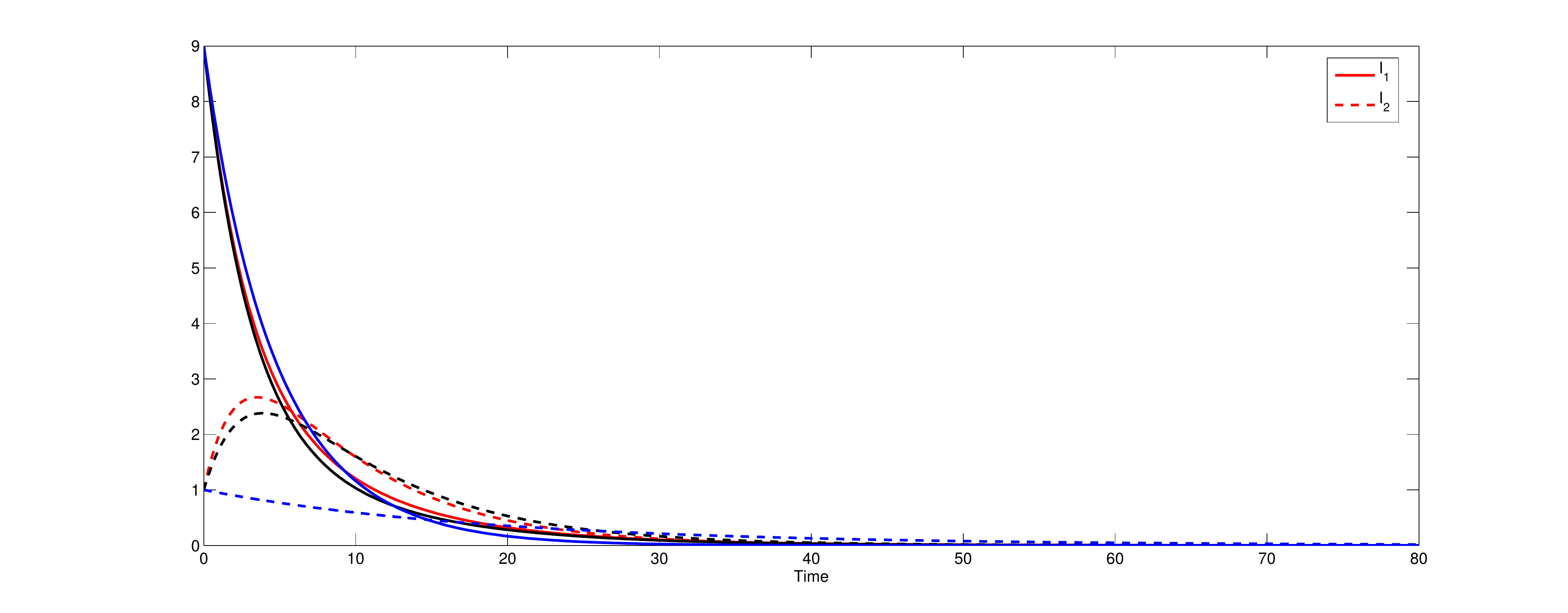}
\caption{\it  Dynamics of $I_1$ and $I_2$ for varying $\sigma_{ij}$ for the state-dependent $p_{ij}(I_1,I_2)$ where $\mathcal R_0<1$. This is obtained by using the values $\beta_1=0.2$ and $\beta_2=0.3$. In all the three cases, the disease dies out in both patches. The black curves represent the symmetric case ( $p_{12}=p_{21}=\sigma_{12}=\sigma_{21}=0.5$ ), the blue line represent the polar case ( $p_{12}=p_{21}=\sigma_{12}=\sigma_{21}=0$) and red curves represent high mobility case ($p_{12}=p_{21}=\sigma_{12}=\sigma_{21}=1)$.}
\label{SDR0Less1}
\end{figure}

\item In Fig ~\ref{I1Sigmas2} and ~\ref{I2Sigmas2}, we explore the cases where there is  symmetry ($\sigma_{ij}=\sigma_{ji}$) with $\sigma_{ij}=p_{ij}(0,0)$. We supposed in this case that Patch 2 has higher risk ($\beta_2=1.2$) and Patch 1 has lower risk ($\beta_1=0.3$). As can be intuitively deduced, the prevalence in Patch 1 is at its highest in the case of ``high mobility" ($\sigma_{12}=\sigma_{21}=1$), and  decreasing as $\sigma_{ij}$ decreases (with $i\neq j$). Conversely,  prevalence in Patch 2 is at its highest under very ``low mobility" ($\sigma_{12}=\sigma_{21}=0$) and decreases as $\sigma_{ij}$ increases. Note that $\sigma_{ij}$, with $i\neq j$ is proportional to $p_{ij}(I_1,I_2)$ which is the actual residence time.\\

\begin{figure}[ht]
\centering
 \subfigure[The level of prevalence in Patch 1 (low risk) seems to decrease as $\sigma_{12}$ and $\sigma_{21}$ decrease.]{
   \includegraphics[scale =.28] {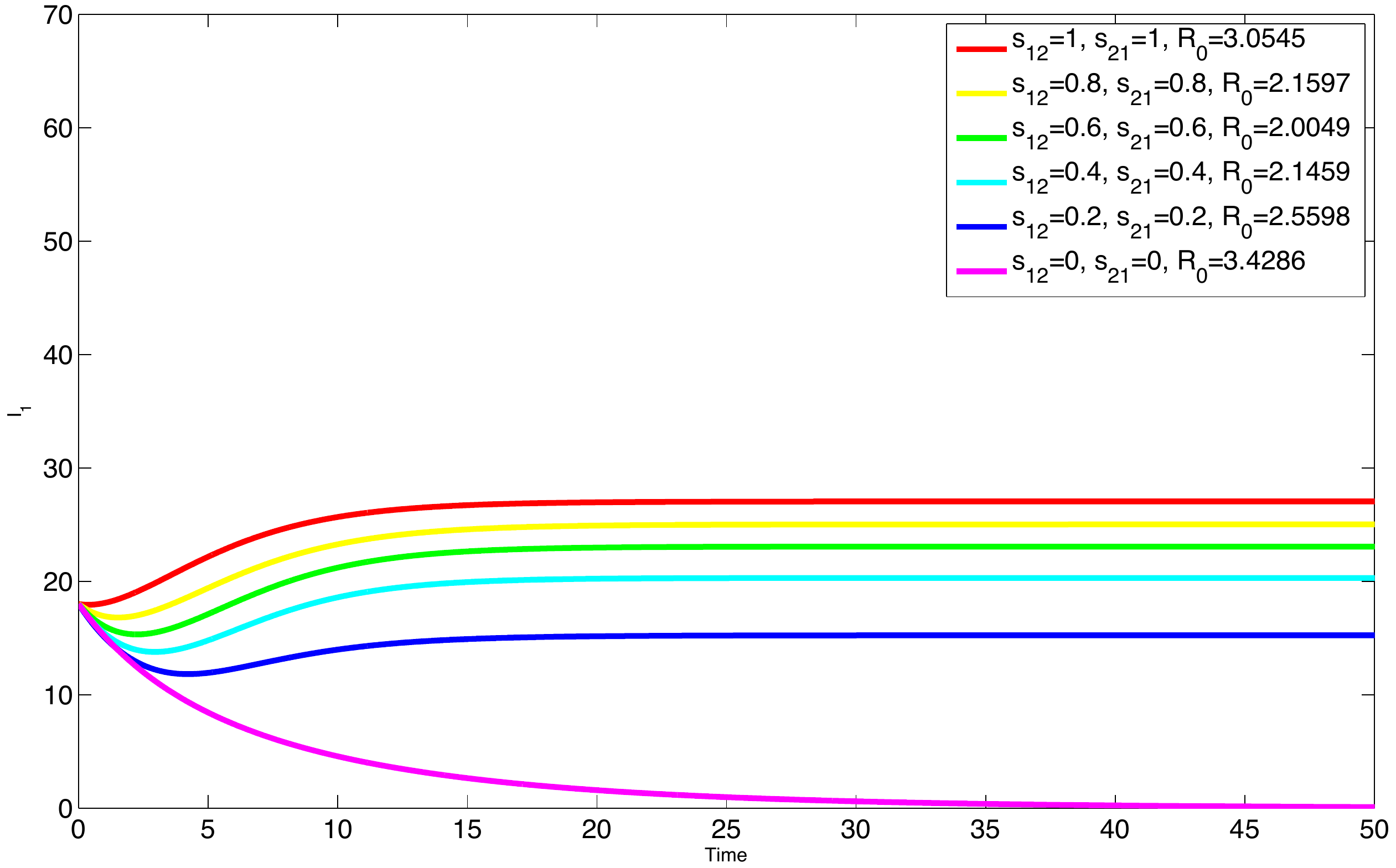}
\label{I1Sigmas2}}
 \subfigure[The level of prevalence in Patch 2 (high risk) seems to increase as $\sigma_{12}$ and $\sigma_{21}$ decrease.]{
   \includegraphics[scale =.28] {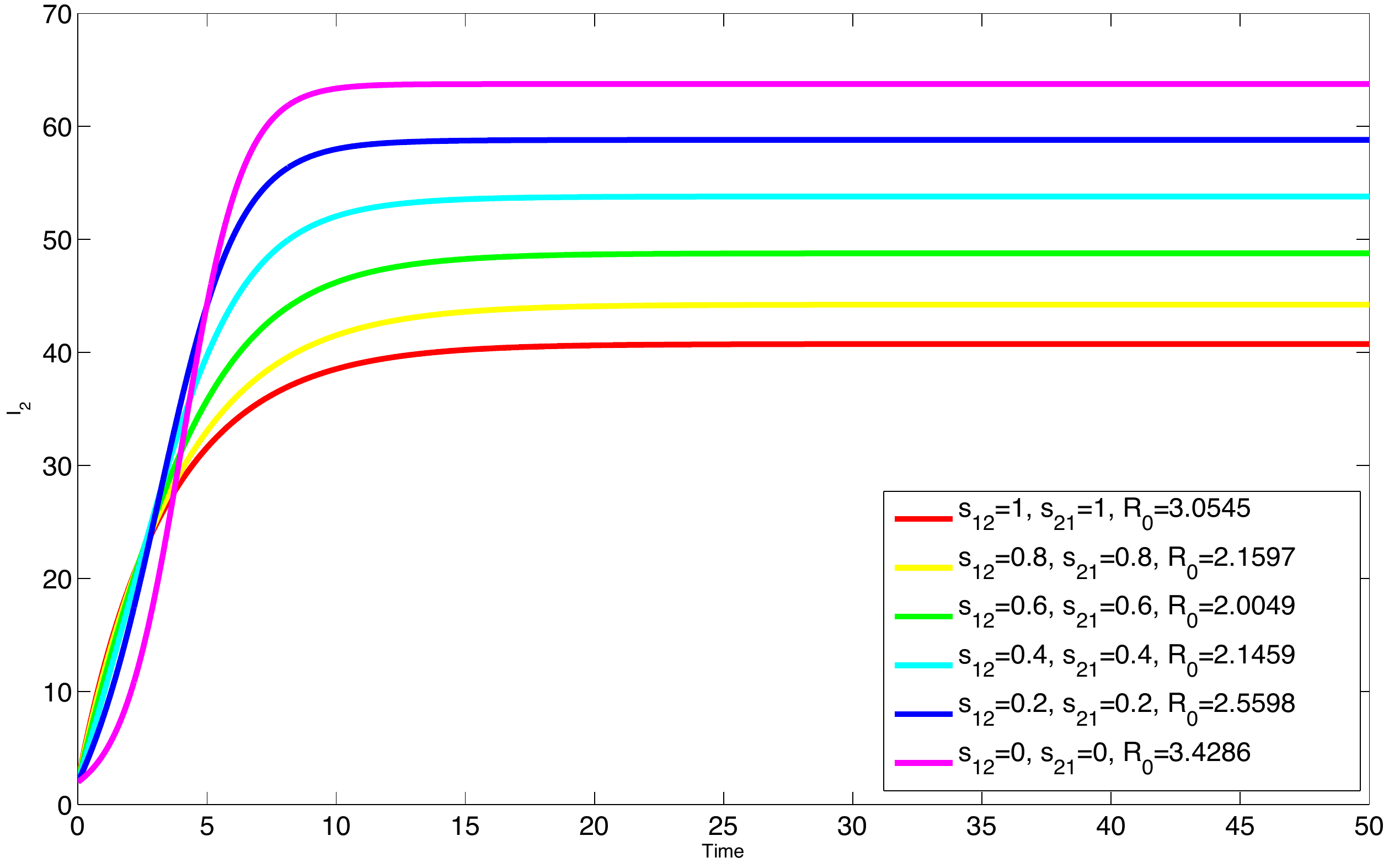}
\label{I2Sigmas2}}
\caption{Dynamics of $I_1$ and $I_2$ for varying $\sigma_{ij}$ for the state-dependent $p_{ij}(I_1,I_2)$ approach.} \label{fig:twofigs}
\end{figure}

\item We continue to explore the asymmetric case ($\sigma_{ij}\neq\sigma_{ji}$), that is, there is more mobility towards one patch. In Fig ~\ref{I1SigmasInverted22},  the prevalence in Patch 1 (low risk) is at its highest if there is ``high mobility" from Patch 1 to Patch 2 ($\sigma_{12}=1$) and no mobility from Patch 2 to Patch 1 ($\sigma_{21}=0$), the prevalence decreases along  with $\sigma_{12}$. If the programmed response of residents of Patch 1 is to reduce their mobility ($\sigma_{12}=0$) then, even if the mobility of residents in the high risk Patch 2 is extremely high ($\sigma_{21}=1$), still the prevalence in Patch 1 is at its lowest. Similar remarks hold for Fig ~\ref{I2SigmasInverted22} regarding the prevalence in Patch 2 (high risk) under different mobility schemes.  \\

\begin{figure}[ht!]
\centering
\subfigure[The level of prevalence in Patch 1 (low risk) seems to decrease as $\sigma_{12}$ and $\sigma_{21}$ decrease.]{
   \includegraphics[scale =.28] {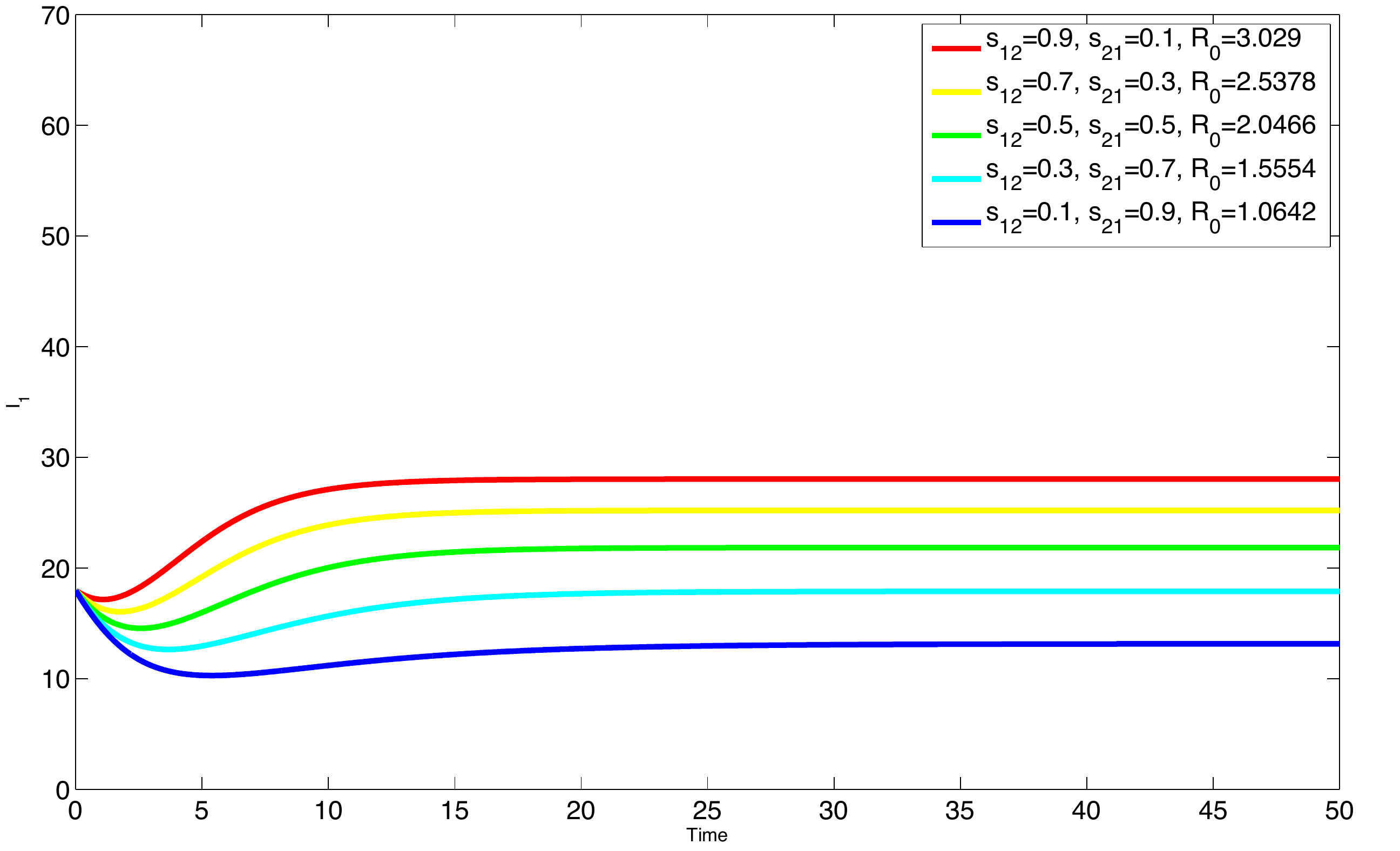}
\label{I1SigmasInverted22}}
  \subfigure[The level of prevalence in Patch 2 (high risk) seems to increase as $\sigma_{12}$ and $\sigma_{21}$ decrease.]{
   \includegraphics[scale =.28] {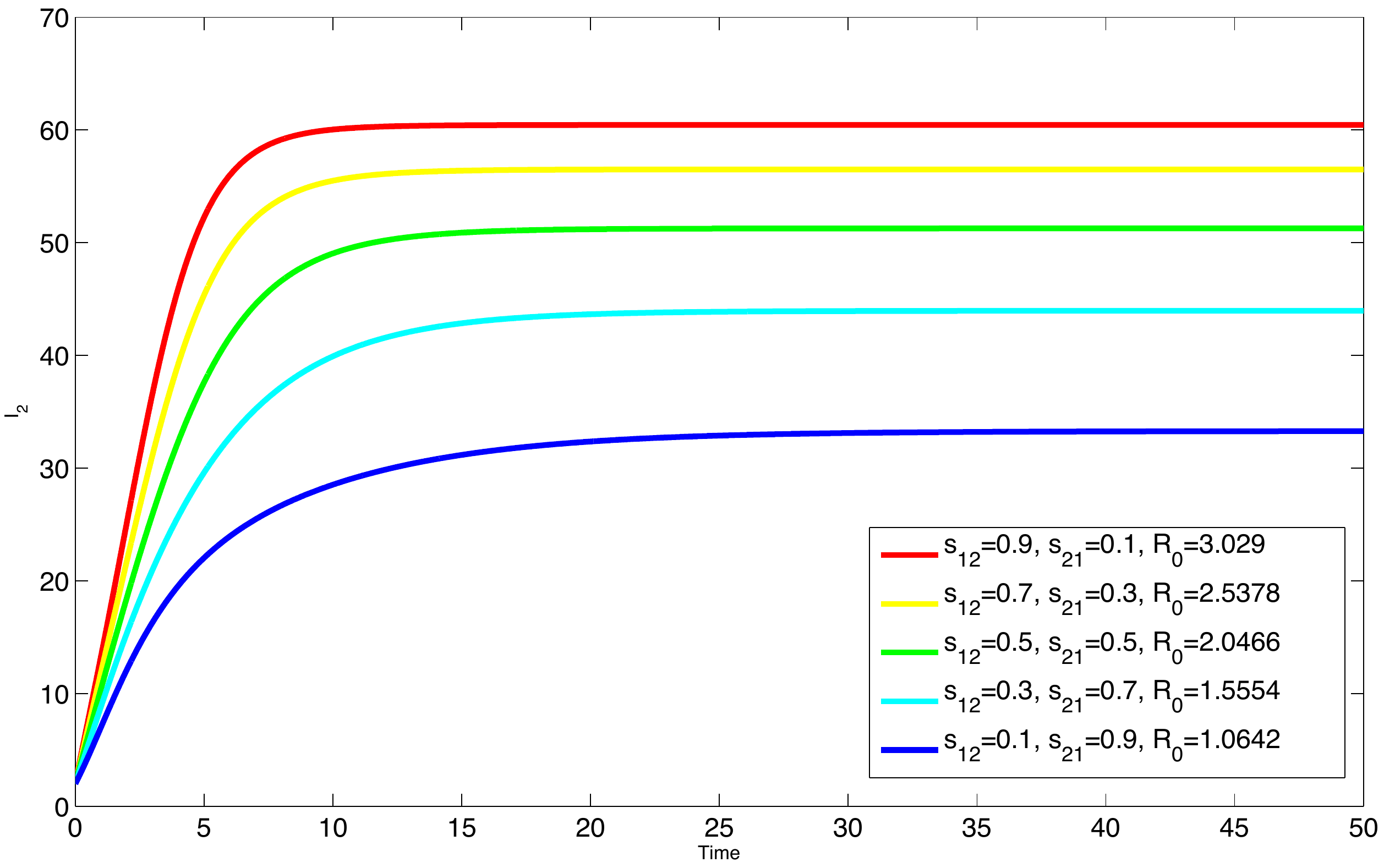}
\label{I2SigmasInverted22}}
   \caption{Dynamics of and $I_1$ and $I_2$ for varying $\sigma_{ij}$, but non-symmetric, for the state-dependent $p_{ij}(I_1,I_2)$.} \label{fig:twofigs}
\end{figure}

\begin{figure}[ht!]
\begin{center}
\includegraphics[scale=0.4]{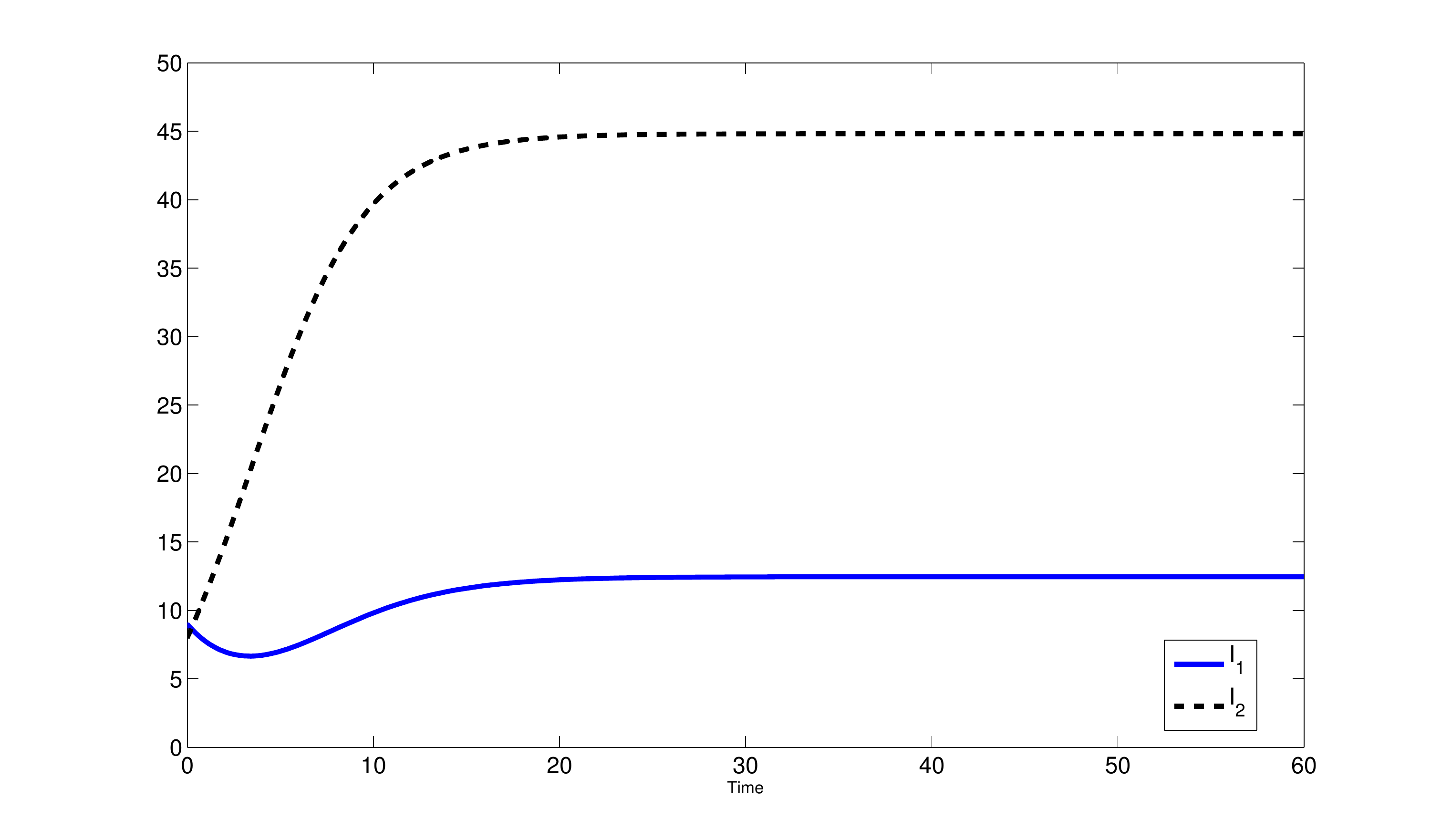}
\caption{\it Dynamics of $I_1$ and $I_2$ where $p_{12}=0$. In this case the residence time matrix $\mathbb P$ is not irreducible, the disease in Patch 2 persists nonetheless as predicted by the theorem \ref{MainTheo2}.}
\label{ReducibleI1I2} 
\end{center}
\end{figure}
\item Finally, Figure \ref{ReducibleI1I2} presents the dynamics of the infected in both patches for the (conventional) case where $p_{12}=0$ ( and $p_{11}=1$). This case is particularly interesting since the residence time matrix $\mathbb P$ is not irreducible (hence the hypothesis of Theorem \ref{MainTheo} fails) but ${\mathcal R}_0^2(\mathbb P)=1.8929>1$. As predicted by Theorem \ref{MainTheo2}, the disease in Patch 2 is persistent. Also, it worth noticing that in Fig \ref{ReducibleI1I2}, $I_1$ persists as well even though ${\mathcal R}_0^1(\mathbb P)=0.4455<1$, as the condition ${\mathcal R}_0^i(\mathbb P)>1$, for $i=1,2$, is sufficient but not necessary for persistence in Patch $i$.
\end{enumerate}

\section{Final epidemic size}\label{sec:FES}
Although the disease dynamics described here are not those of a controlled epidemiological system (the $\mathcal R_0$ is that corresponding to an uncontrolled system) they are still of considerable interest. The study of the role of residence time matrices on the dynamics of a single outbreak within a Susceptible-Infected-Recovered (with immunity) or SIR model without births and deaths is relevant to the development of public disease management measures \cite{BrauerFengCCC2010,ChowellCCCetAL2015,NancyFengCCC2013}. Under the parameters and definitions introduced earlier, and making use of the same notation, we arrive at the following system of nonlinear differential equations:
 \begin{equation} \label{SIROutbreat}
\left\{\begin{array}{ll}
\dot S_i=-\left(\frac{\beta_ip^2_{ii}}{p_{ii}N_i+p_{ji}N_j}+\frac{\beta_jp^2_{1ij}}{p_{ij}N_i+p_{jj}N_j}\right)S_iI_i    -  \left(\frac{\beta_ip_{ii}p_{ji}}{p_{ii}N_i+p_{ji}N_j}+\frac{\beta_jp_{ij}p_{jj}}{p_{ij}N_i+p_{jj}N_j}\right)S_iI_j,\\\\
\dot I_i=\left(\frac{\beta_ip^2_{ii}}{p_{ii}N_i+p_{ji}N_j}+\frac{\beta_jp^2_{1ij}}{p_{ij}N_i+p_{jj}N_j}\right)S_iI_i  + \left(\frac{\beta_ip_{ii}p_{ji}}{p_{ii}N_i+p_{ji}N_j}+\frac{\beta_jp_{ij}p_{jj}}{p_{ij}N_i+p_{jj}N_j}\right)S_iI_j            -\alpha_i I_i,\\\\
\dot R_i=\alpha_i I_i,
\end{array}\right.
\end{equation}where $R_i$ denotes the population of recovered immune individuals in Patch $i$, $\alpha_i$ is the recovery rate in Patch $i$ and $N_i\equiv S_i+I_i+R_i$, for $i=1,2$.\\

The basic reproduction number $\mathcal R_0$, is by definition the largest eigenvalue of $2\times2$ ($n\times n$ for the general case) next generation matrix,
{\small{
$$-FV^{-1}=
\left(\begin{array}{cc}
\left(\frac{\beta_1p^2_{11}}{p_{11}N_1+p_{21}N_2}+\frac{\beta_2p^2_{12}}{p_{12}N_1+p_{22}N_2}\right)\frac{N_1}{\alpha_1} & \left(\frac{\beta_1p_{11}p_{21}}{p_{11}N_1+p_{21}N_2}+\frac{\beta_2p_{12}p_{22}}{p_{12}N_1+p_{22}N_2}\right)\frac{N_1}{\alpha_2} \\
\left(\frac{\beta_1p_{11}p_{21}}{p_{11}N_1+p_{21}N_2}+\frac{\beta_2p_{12}p_{22}}{p_{12}N_1+p_{22}N_2}\right)\frac{N_2}{\alpha_1} & \left(\frac{\beta_1p^2_{21}}{p_{11}N_1+p_{21}N_2}+\frac{\beta_2p^2_{22}}{p_{12}N_1+p_{22}N_2}\right)\frac{N_2}{\alpha_2} 
\end{array}\right).
$$
}}

It has been shown (see \cite{Hethcote76}, for example) that not everybody gets infected during an outbreak, and so, estimating the size of the recovered population (the final epidemic size in the absence of deaths or departures) is tied in the solutions of the final size relationship, given in this case, by the system:
{
\small{
\begin{equation}\label{FSEMatrix}
\begin{bmatrix}
   \log\frac{S_1(0)}{S_1(\infty)}\\\\
\log\frac{S_2(0)}{S_2(\infty)} 
  \end{bmatrix}
=
\begin{bmatrix}
K_{11} & K_{12} \\
\\
K_{21}& K_{22}
 \end{bmatrix}
\begin{bmatrix}
1-\frac{S_1(\infty)}{N_1}\\\\
1-\frac{S_2(\infty)}{N_2} 
  \end{bmatrix}
\end{equation}
}}
where 

$$K=\begin{bmatrix}
\left(\frac{\beta_1p^2_{11}}{p_{11}N_1+p_{21}N_2}+\frac{\beta_2p^2_{12}}{p_{12}N_1+p_{22}N_2}\right)\frac{N_1}{\alpha_1} & \left(\frac{\beta_1p_{11}p_{21}}{p_{11}N_1+p_{21}N_2}+\frac{\beta_2p_{12}p_{22}}{p_{12}N_1+p_{22}N_2}\right)\frac{N_2}{\alpha_2} \\
\\
\left(\frac{\beta_1p_{11}p_{21}}{p_{11}N_1+p_{21}N_2}+\frac{\beta_2p_{12}p_{22}}{p_{12}N_1+p_{22}N_2}\right)\frac{N_1}{\alpha_1} & \left(\frac{\beta_1p^2_{21}}{p_{11}N_1+p_{21}N_2}+\frac{\beta_2p^2_{22}}{p_{12}N_1+p_{22}N_2}\right)\frac{N_2}{\alpha_2} 
 \end{bmatrix}.\vspace{5pt}$$

The relationship (\ref{FSEMatrix}) is obtained by using the fact that, in (\ref{SIROutbreat}), we have $\dot S_i+\dot I_i=-\alpha_i I_i\geq 0$. This implies that $\lim_{t\to\infty}I_i(t)=0$ (for $i=1, 2$), since $S_i$ and $I_i$ are positive and integrating $\frac{\dot S_i}{S_i}$ in (\ref{SIROutbreat}), we obtain, after some tedious algebra Expression (\ref{FSEMatrix}). The references \cite{Brauer2008,BrauerCCC2012} give more details on the computation of the final size relationship. \\

It is important to observe that the next generation matrix and the matrix $K$ defining the final epidemic size have the same eigenvalues. And so, the dominant eigenvalue, for both is $\mathcal R_0$ (although we note that we would not expect this to be the case in a controlled epidemiological system).\\

The residence time matrix $\mathbb P$ plays an important role as evidenced  by the dependence of the final epidemic size relation as in Fig \ref{FES}.  As we can notice in Fig \ref{FES}, the prevalence in low risk Patch 1 is highest in the high mobility scheme where as in high risk Patch 2, the high mobility leads to the lowest prevalence. Also, as stated before ($\displaystyle \lim_{t\to+\infty}I_i(t)=0$, for $i=1,\; 2$.) with any typical outbreak model, the disease ultimately dies out from both patches \cite{HethcoteSIAM2000}.

\begin{figure}[h!]
\includegraphics[scale=0.45]{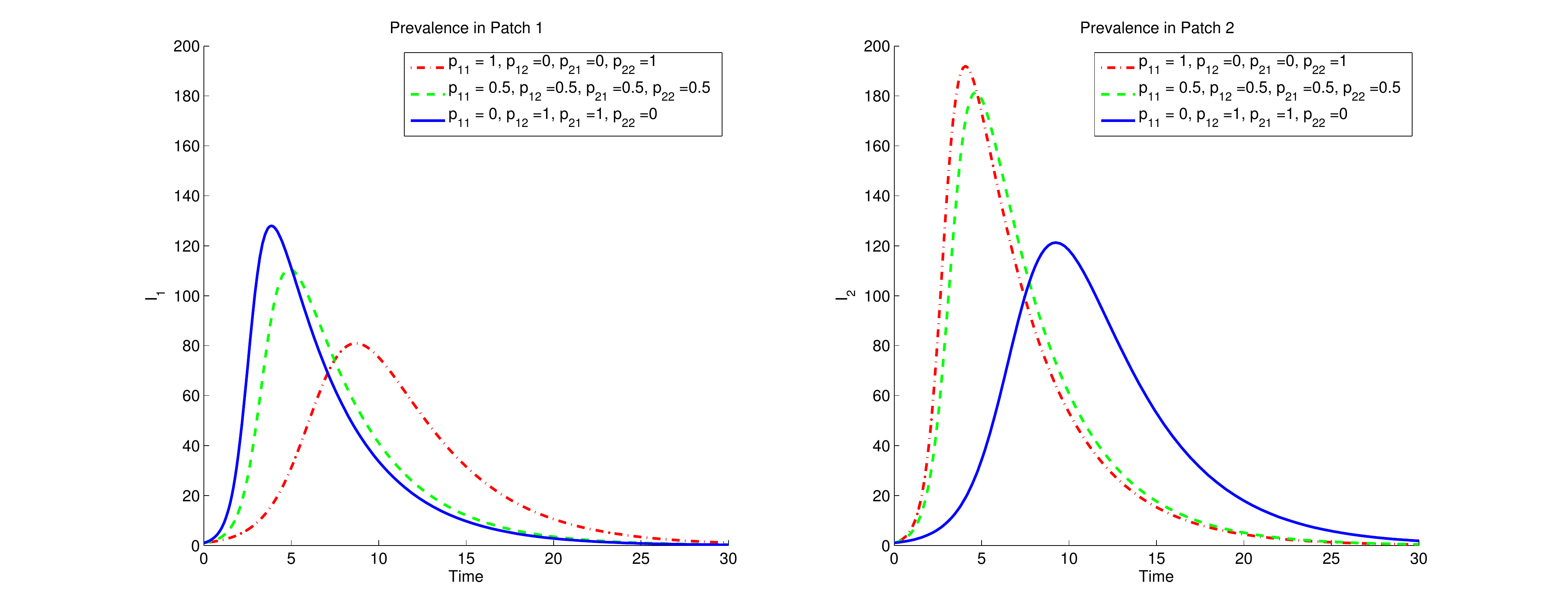}
\caption{ The prevalence in patch 1 (low risk) reaches its highest when in extreme mobility case (solid blue line) and is lowest when there is no mobility between the patches. The opposite of this scenario happens in patch 2 (high risk).}
\label{FES}
\end{figure}

%
\section{Conclusion and Discussions}\label{CCLDiscussions}
Heterogeneous mixing in multi-group epidemic models is most often defined in terms of group specific susceptibility and average contact rates captured multiplicatively by the transmission parameter $\beta$. However, contact rates, in general, cannot be measured in satisfactory ways for diseases like influenza, measles or tuberculosis, due to the difficulty of assessing the average number of contacts per unit of time of susceptible populations in different locations for varied activities. In this paper we propose the use of residence times in heterogeneous environments, as a proxy for ``effective" contacts over a certain time window; and develop a multi-group epidemic framework via virtual dispersal where the risk of infection is  a function of the residence time and local environmental risk. This novel approach eliminates the need to define and measure contact rates that are used in the traditional multi-group epidemic models with heterogeneous mixing.\\

Under the proposed framework, we formulate a general multi-patch $SIS$ epidemic model with residence times. We calculate the basic reproduction number $\mathcal R_0 $ which is a function of a patch residence-times matrix $\mathbb P$. Our global analysis shows that the model is robust in the sense that the disease dynamics depend exclusively on the basic reproductive number when the residence times matrix $\mathbb P$ is ``constant" (Theorem \ref{MainTheo}). We proved that the disease free equilibrium is globally asymptotically stable (GAS) if the basic reproduction number $\mathcal R_0\leq1$ and that a unique interior endemic GAS equilibrium exists if $\mathcal R_0>1$. This results holds as long as the residence time matrix $\mathbb P$ is irreducible, that is, the graph of the patches is strongly connected. \\

Our further analysis (Theorem \ref{MainTheo2}) provide  easily accessible insights on the impact of the residence matrix $\mathbb P$ on the levels of infection within each patch. Our results imply that  the infection risk (measured by $\mathcal B$) and the residence time matrix ($\mathbb P$) can play an important role in the endemic at the patch level. More specifically, the right combinations of the environmental risk level ($\mathcal B$) and dispersal behavior ($\mathbb P$) can either promote or suppress infection for particular patches. This work complements the results of Theorem \ref{MainTheo} regarding the robust dynamics under the assumption that $\mathbb P$ is strongly connected, i.e., irreducible. For example, when Theorem \ref{MainTheo2} is applied to the two patch case,  residents of Patch 1 visit  Patch 2 but not conversely.

These significant differences that emerges from the study of residence times models ($\mathbb P$ a ``constant") includes the possibility of studying disease dynamics in non strongly connected pacha configuration. In particular, we found conditions that allow us to characterize the patch-specific disease dynamics as a function of the time spend by residents and visitors to the patch of interest. This approach allowed us to classify patches as sources or sinks of infection, a role that depends on risk ($\mathcal B$) and mobility ($\mathbb P$).\\

 We also explored the case where the entries of residence times matrix $\mathbb P$ are no longer constant but rather prevalence dependent. We noticed that whenever the residence times are negatively correlated with risk then prevalence will be higher in the riskier patch but much lower than if the residence times were independent of health status.
 We ran carefully designed simulations to gain insights on the use of phenomenological modeling approach (System (\ref{2PNI3Variable})), since the mathematical analysis would be in general challenging. \\

Our proposed framework has been applied to the context of a two-patch $SIR$ single outbreak model with the concept of residence times. We derived the final epidemic size relationship in order to capture the size of the outbreak. Our results show that the residence time matrix $\mathbb P$ plays an important role which evidenced  by the dependence of the final epidemic size relation as in Fig \ref{FES}.\\

In both conventional and phenomenological approaches to residence times used in this paper, humans behavior and responses to disease risk are automatic: $\mathbb P$ is constant and predefined functions of health status. Recent studies \cite{EliCCC2011,Horan2007,Horan2011,Horan2010,Perrings:2014yq} have incorporated behavior as a feedback response coupled with the dynamics of the disease. A model of the decision to spend time in patch $i=1,2$ based on individuals' utility functions that include the possibility of adapting to changing contagion dynamics in the above two patch setting, using previous work \cite{EliCCC2011,MorinCCC2013}, is the subject of a separate study.
\section*{Acknowledgements}
These studies were made possible by grant \#1R01GM100471-01 from the National Institute of General Medical Sciences (NIGMS) at the National Institutes of Health. The contents of this manuscript are solely the responsibility of the authors and do not necessarily represent the official views of DHS or NIGMS. Research of Y.K. is partial supported by NSF-DMS (1313312). The funders had no role in study design, data collection and analysis, decision to publish, or preparation of the manuscript.
\appendix

\section{Computation of $\mathcal R_0$}\label{R0}
\begin{proof}
The general SIS model with residence time is described by the system (\ref{SIScompact})
$$
\dot I=\textrm{diag}(\bar N-I)\mathbb P\textrm{diag}(\mathcal B)\textrm{diag}(\tilde N)^{-1}\mathbb P^tI-\textrm{diag}(d_I+\gamma_I)I.
$$
The right hand member of the above system be can clearly decomposed as $\mathcal F+\mathcal V$ where
$$\mathcal F=\textrm{diag}(\bar N-I)\mathbb P\textrm{diag}(\mathcal B)\textrm{diag}(\tilde N)^{-1}\mathbb P^tI
 \quad\textrm{and}\quad \mathcal V=-\textrm{diag}(d_I+\gamma_I)I$$
 
 The jacobian at the DFE of $\mathcal F$ and $\mathcal V$ are giving by:
 
 $$F=D\mathcal F\bigg|_{DFE}=\textrm{diag}(\bar N)\mathbb P\textrm{diag}(\mathcal B)\textrm{diag}(\tilde N)^{-1}\mathbb P^t \quad\textrm{and}\quad V=\mathcal V\bigg|_{DFE}=-\textrm{diag}(d_I+\gamma_I)$$
 
The basic reproduction number $\mathcal R_0$ is given by the spectral radius of the next generation matrix $-FV^{-1}$  \cite{MR1057044,VddWat02}. Hence, we deduce that

$$\mathcal R_0=\rho(-\textrm{diag}(\bar N)\mathbb P\textrm{diag}(\mathcal B)\textrm{diag}(\tilde N)^{-1}\mathbb P^t V^{-1} )$$

For the two-patch SIS model (\ref{SIScompact}), $\mathcal R_0$ is the largest eigenvalue of the following matrix:
$$-FV^{-1}=\left(\begin{array}{cc}
\frac{b_1}{d_1(d_1+\gamma_1)}\left(\frac{\beta_1p^2_{11}}{p_{11}\frac{b_1}{d_1}+p_{21}\frac{b_2}{d_2}}+\frac{\beta_2p^2_{12}}{p_{12}\frac{b_1}{d_1}+p_{22}\frac{b_2}{d_2}}\right) & 
\frac{b_1}{d_1(d_2+\gamma_2)}\left(\frac{\beta_1p_{11}p_{21}}{p_{11}\frac{b_1}{d_1}+p_{21}\frac{b_2}{d_2}}+\frac{\beta_2p_{12}p_{22}}{p_{12}\frac{b_1}{d_1}+p_{22}\frac{b_2}{d_2}}\right)\\
\frac{b_2}{d_2(d_1+\gamma_1)} \left(\frac{\beta_1p_{11}p_{21}}{p_{11}\frac{b_1}{d_1}+p_{21}\frac{b_2}{d_2}}+\frac{\beta_2p_{12}p_{22}}{p_{12}\frac{b_1}{d_1}+p_{22}\frac{b_2}{d_2}}\right)& 
\frac{b_2}{d_2(d_2+\gamma_2)}\left(\frac{\beta_1p^2_{21}}{p_{11}\frac{b_1}{d_1}+p_{21}\frac{b_2}{d_2}}+\frac{\beta_2p^2_{22}}{p_{12}\frac{b_1}{d_1}+p_{22}\frac{b_2}{d_2}}\right)
\end{array}\right)$$
Let $$\heartsuit=\frac{b_1}{d_1(d_1+\gamma_1)}\left(\frac{\beta_1p^2_{11}}{p_{11}\frac{b_1}{d_1}+p_{21}\frac{b_2}{d_2}}+\frac{\beta_2p^2_{12}}{p_{12}\frac{b_1}{d_1}+p_{22}\frac{b_2}{d_2}}\right)$$
$$\diamondsuit=\frac{b_1}{d_1(d_2+\gamma_2)}\left(\frac{\beta_1p_{11}p_{21}}{p_{11}\frac{b_1}{d_1}+p_{21}\frac{b_2}{d_2}}+\frac{\beta_2p_{12}p_{22}}{p_{12}\frac{b_1}{d_1}+p_{22}\frac{b_2}{d_2}}\right)$$

$$\clubsuit=\frac{b_2}{d_2(d_1+\gamma_1)} \left(\frac{\beta_1p_{11}p_{21}}{p_{11}\frac{b_1}{d_1}+p_{21}\frac{b_2}{d_2}}+\frac{\beta_2p_{12}p_{22}}{p_{12}\frac{b_1}{d_1}+p_{22}\frac{b_2}{d_2}}\right)$$
and 

$$\spadesuit=\frac{b_2}{d_2(d_2+\gamma_2)}\left(\frac{\beta_1p^2_{21}}{p_{11}\frac{b_1}{d_1}+p_{21}\frac{b_2}{d_2}}+\frac{\beta_2p^2_{22}}{p_{12}\frac{b_1}{d_1}+p_{22}\frac{b_2}{d_2}}\right)$$
Then,

$$\mathcal R_0=\frac{1}{2}\left(\heartsuit+\spadesuit+\sqrt{(\heartsuit+\spadesuit)^2-4(\heartsuit\spadesuit-\diamondsuit\clubsuit)}\right)$$
With$$\clubsuit=\frac{b_2}{d_2(d_1+\gamma_1)} \frac{d_1(d_2+\gamma_2)}{b_1}\diamondsuit$$
\end{proof}
\section{Proof of Theorem \ref{MainTheo}}\label{AppProof}

The proof uses the method in \cite{IggidrSalletTsanou} which is based on Hirsch's theorem  \cite{Hirsch84}.
\begin{theorem}[Hirsch \cite{Hirsch84}]\hfill

Let $\dot x=F(x)$ be a cooperative differential equation for which $\mathbb R^n_+$ is invariant , the origin is an equilibrium, each $DF(x)$ is irreducible,  and that all orbits  are bounded. Suppose that $$x>y\implies DF(x)<DF(y)\quad\textrm{for all}\quad x,y.$$
Then all orbits in $\mathbb R^n_+$ tend to zero or there is a unique equilibrium $p^\ast$ in the interior of $\mathbb R^n_+$ and all orbits in $\mathbb R^n_+$ tend to $p^\ast$.
\end{theorem}
\begin{proof}[Proof of Theorem \ref{MainTheo}] \hfill

Equation (\ref{SIScompact}) can be written as:
\begin{equation}\label{SIScompactF}\dot I=(F+V)I-\textrm{diag}(I)\mathbb P\textrm{diag}(\mathcal B)\textrm{diag}(\tilde N)^{-1}\mathbb P^tI\end{equation}
where $F=\text{diag}(\bar N)\mathbb P\text{diag}(\mathcal B)\text{diag}(\tilde N)^{-1}\mathbb P^t$ and $V=-\text{diag}(d_I+\gamma_I)$, as defined in Appendix \ref{R0}. Let us denote by $X(I)$ the semi flow induced by (\ref{SIScompactF}). Hence
\begin{equation}\label{DX}
DX(I)=\textrm{diag}(\bar N-I) \mathbb P\textrm{diag}(\mathcal B)\textrm{diag}(\tilde N)^{-1}\mathbb P^t+V-W(I_1,I_2)
\end{equation}
where $W(I_1,I_2)=\textrm{diag}(\mathbb P\textrm{diag}(\mathcal B)\textrm{diag}(\tilde N)^{-1}\mathbb P^t I)$. Since $\mathbb P$ is irreducible and $I\leq \bar N$, $DX(I)$ is clearly Metzler irreducible matrix. That means, the flow is strongly monotone. Plus, $DX(I)$ is clearly decreasing with respect of $I$. Hence, by Hirsch's theorem either all trajectories go to zero or go to an equilibrium point $\bar I\gg0$.
From the relation (\ref{DX}), we have $DX(0)=F+V$ where $F$ and $V$ are the one defined previously in Appendix \ref{R0}.  However, since $F$ a nonnegative matrix and $V$ is Metzler, we have the following equivalence

$$\alpha(F+V)<0\iff \rho(-FV^{-1})<1$$

where $\alpha(F+V)$ is the stability modulus, i.e: the largest real part of eigenvalues, of $F+V$ and $\rho(-FV^{-1})$ the spectral radius of $-FV^{-1}$. Hence, the DFE is globally asymptotically stable if $\mathcal R_0=\rho(-FV^{-1})<1$. And if $\mathcal R_0>1$, i.e: $\alpha(F+V)>0$, the DFE is unstable \cite{VddWat02}. Since, we have proved that $DX(I)$ is a Metzler matrix,  to prove the local stability of the endemic equilibrium $\bar I\gg0$, we only need to prove that it exists $w\gg0$ such that $DX(\bar I)w<0$ \cite{MR1298430}. The endemic equilibrium $\bar I\gg0$ satisfies the equation
$$(F+V)\bar I-\textrm{diag}(\bar I)\mathbb P\textrm{diag}(\mathcal B)\textrm{diag}(\tilde N)^{-1}\mathbb P^t\bar I=0$$
Hence,
$$DX(\bar I)\bar I=-W(\bar I)\bar I<0$$

Hence, with $w=\bar I$, we deduce that $\bar I$ is locally stable. With the attractivity of $\bar I$ guaranteed Hirsh's theorem, we conclude that the endemic equilibrium $\bar I\gg0$ is globally asymptotically stable if $\mathcal R_0>1$.

Finally, if $\mathcal R_0=1$, we have $\alpha(F+V)=0$. It exists $c\gg0$ such that $(F+V)^tc=0$. By considering the Lyapunov function $V=\left\langle c |I\right\rangle$. This function is definite positive and its derivation along the trajectories if (\ref{SIScompactF}) is 
\begin{eqnarray}
\dot V&=&\left\langle c |\dot I\right\rangle\nonumber\\
&=&\left\langle c |(F+V)I-\textrm{diag}(I)\mathbb P\textrm{diag}(\mathcal B)\textrm{diag}(\tilde N)^{-1}\mathbb P^tI\right\rangle\nonumber\\
&=&-\left\langle c |\textrm{diag}(I)\mathbb P\textrm{diag}(\mathcal B)\textrm{diag}(\tilde N)^{-1}\mathbb P^tI\right\rangle\nonumber\\
&\leq&0
\end{eqnarray}
Plus $\dot V=0$ only at the DFE. Hence the DFE is GAS if $\mathcal R_0=1$.  This completes the proof of the theorem \ref{MainTheo}.
\end{proof}

\section{Proof of Theorem \ref{MainTheo2}}\label{AppProof2}
\begin{proof}Since Model (\ref{SIScompact}) has the compact global attractor $\Omega$, then according to Theorem \eqref{MainTheo}, we can expect that $\lim_{t\rightarrow\infty} I_i(t)<\frac{b_i}{d_i}$, thus for time large enough, we can have $\frac{b_i}{d_i}-I_i>0$, therefore we have
$$\begin{array}{lcl}
\dot I_i&>&I_i\left(\frac{b_i}{d_i}-I_i\right)\left(  \sum_{j=1}^{n}\frac{\beta_jp_{ij}^2}{\sum_{k=1}^{n}p_{kj}\frac{b_k}{d_k}}\right) -(d_i+\gamma_i )I_i 
\end{array}$$which indicates follows when $\mathcal R_0^i(\mathbb P)>1 $
$$ \begin{array}{lcl}
\frac{\dot I_i}{I_i}\big\vert_{I_i=0}&=&\frac{b_i}{d_i}\left(  \sum_{j=1}^{n}\frac{\beta_jp_{ij}^2}{\sum_{k=1}^{n}p_{kj}\frac{b_k}{d_k}}\right) -(d_i+\gamma_i )>0 
\end{array}.$$
Then apply the average Lyapunov Theorem \cite{Hutson84}, we can conclude that $\liminf_{t\rightarrow\infty} I_i(t)>0$, i.e., the disease in the residence Patch $i$ is persistent if $\mathcal R_0^i(\mathbb P)>1 $ . \\

If $p_{ij}>0$ and $p_{kj}=0$ for all $k=1,..,n, \mbox{ and } k\neq i$, this implies that if there is a portion of the residence Patch $i$ population flowing into the residence Patch $j$, then there is no other residence Patch $k$ where $k\neq j$, i.e., 
$$\beta_jp_{ij}\sum_{k=1,k\neq i}^{n}p_{kj}I_k=0$$ which also implies that 
$$\left(\frac{b_i}{d_i}-I_i\right)\sum_{j=1}^{n}\frac{\beta_jp_{ij}\sum_{k=1,k\neq i}^{n}p_{kj}I_k}{\sum_{k=1}^{n}p_{kj}\frac{b_k}{d_k}} =0.$$
then we can conclude that  Model  (\ref{SIScompact}) can have an equilibrium  since under these conditions,
$$\frac{b_i}{d_i}\sum_{j=1}^{n}\frac{\beta_jp_{ij}\sum_{k=1,k\neq i}^{n}p_{kj}I_k}{\sum_{k=1}^{n}p_{kj}\frac{b_k}{d_k}}=\frac{b_i}{d_i}\frac{\beta_i \sum_{k=1,k\neq i}^{n}p_{ki}I_k}{\sum_{k=1}^{n}p_{kj}\frac{b_k}{d_k}} =0.$$ Therefore, if the conditions $p_{kj}=0$ for all $k=1,..,n, \mbox{ and } k\neq j$ whenever $p_{ij}>0$ hold, then we have 
$$\begin{array}{lcl}
\dot I_i\vert_{I_i=0}&=&\left[I_i\left(\frac{b_i}{d_i}-I_i\right)\left(  \sum_{j=1}^{n}\frac{\beta_jp_{ij}^2}{\sum_{k=1}^{n}p_{kj}\frac{b_k}{d_k}}\right)+\left(\frac{b_i}{d_i}-I_i\right)\sum_{j=1}^{n}\frac{\beta_jp_{ij}\sum_{k=1,k\neq i}^{n}p_{kj}I_k}{\sum_{k=1}^{n}p_{kj}\frac{b_k}{d_k}}  -(d_i+\gamma_i )I_i \right]\bigg\vert_{I_i=0} \\
&=&\frac{b_i}{d_i}\sum_{j=1}^{n}\frac{\beta_jp_{ij}\sum_{k=1,k\neq i}^{n}p_{kj}I_k}{\sum_{k=1}^{n}p_{kj}\frac{b_k}{d_k}}=0
\end{array}.$$Therefore, $I_i=0$ is the invariant manifold for Model (\ref{SIScompact}).

On the other hand, when these conditions hold, then we have 
$$\mathcal R_0^i(\mathbb P)=R_0^i \times \sum_{j=1}^n\left(\frac{\beta_j}{\beta_i}\right) p_{ij}\left( \frac{\left(p_{ij} \frac{b_i}{d_i}\right)}{\sum_{k=1}^{n}p_{kj}\frac{b_k}{d_k}}\right)=R_0^i \times \sum_{j=1}^n \left(\frac{\beta_j}{\beta_i}\right) p_{ij}.$$ Therefore, if 
$\mathcal R_0^i(\mathbb P)=R_0^i \times\sum_{j=1}^n \left(\frac{\beta_j}{\beta_i}\right) p_{ij}<1$, then we have the following inequality:
$$ \begin{array}{lcl}
\frac{\dot I_i}{I_i}&=&I_i\left(\frac{b_i}{d_i}-I_i\right)\left(  \sum_{j=1}^{n}\frac{\beta_jp_{ij}^2}{\sum_{k=1}^{n}p_{ki}\frac{b_k}{d_k}}\right)-(d_i+\gamma_i )I_i\\
&\leq& I_i \left[\frac{b_i}{d_i}\left(  \sum_{j=1}^{n}\frac{\beta_jp_{ij}^2}{\sum_{k=1}^{n}p_{ki}\frac{b_k}{d_k}}\right)-(d_i+\gamma_i )\right]\\
&=&I_i \left[\sum_{j=1}^{n}\beta_jp_{ij}-(d_i+\gamma_i )\right]<0
\end{array}.$$
 Therefore, we have $\lim_{t\rightarrow\infty} I_i(t)=0$, i.e., there is no endemic in the residence Patch $i$.
 \end{proof}


\end{document}